\documentclass[12pt, draftclsnofoot, onecolumn]{IEEEtran}
\ifCLASSINFOpdf
\else
\fi
\usepackage{array}
\usepackage[hang]{footmisc}

\usepackage[subrefformat=parens,labelformat=parens,caption=false,font=footnotesize]{subfig}
\usepackage{amsmath}
\usepackage[nolist,nohyperlinks]{acronym}
\usepackage{amsmath}

	\begin{acronym}
    \acro{4G}{fourth generation}
    \acro{5G}{fifth generation}
    \acro{AoA}{angle of arrival}
   	\acro{AoD}{angle of departure}
    \acro{BCRLB}{Bayesian CRLB}
    \acro{BS}{base stations}
	\acro{CDF}{cumulative density function}
    \acro{CF}{closed-form}
    \acro{CRLB}{Cramer-Rao lower bound}
    \acro{EI}{Exponential Integral}
    \acro{eMBB}{enhanced mobile broadband}
    \acro{FIM}{Fisher Information Matrix}
    \acro{GPS}{global positioning system}
    \acro{GNSS}{global navigation satellite system}
    \acro{HetNets}{heterogeneous networks}
    \acro{LOS}{line of sight}
    \acro{MBS}{macro base station}
    \acro{MIMO}{multiple input multiple output}
    \acro{mm-wave}{millimeter wave}
    \acro{mMTC}{massive machine-type communications}
    \acro{MS}{mobile station}
    \acro{MVUE}{minimum-variance unbiased estimator}
    \acro{NLOS}{non line-of-sight}
    \acro{OFDM}{orthogonal frequency division multiplexing}
    \acro{PDF}{probability density function}
    \acro{PGF}{probability generating functional}
    \acro{PLCP}{Poisson line Cox process}
    \acro{PLT}{Poisson line tessellation}
    \acro{PLP}{Poisson line process}
    \acro{PPP}{Poisson point process}
    \acro{PV}{Poisson-Voronoi}
    \acro{QoS}{quality of service}
    \acro{RAT}{radio access technique}
    \acro{RSSI}{received signal-strength indicator}
    \acro{BS}{small cell base station}
   	\acro{SINR}{signal to interference plus noise ratio}
    \acro{SNR}{signal to noise ratio}
	\acro{ULA}{uniform linear array}
	\acro{UE}{user equipment}
 	\acro{URLLC}{ultra-reliable low-latency communications}
    \acro{V2V}{vehicle-to-vehicle}    
\end{acronym}

\usepackage{graphicx,wrapfig,lipsum}
\usepackage{rotating}
\usepackage{amsmath, amssymb, amsthm}
\usepackage{stmaryrd}
\usepackage{tikz}
\usepackage{verbatim}
\usepackage{url}
\usepackage[utf8]{inputenc}
\usepackage[english]{babel}
\newtheorem{theorem}{Theorem}[]
\newtheorem{corollary}{Corollary}[]
\newtheorem{proposition}{Proposition}[]
\newtheorem{lemma}[]{Lemma}
\usepackage{multicol}
\usepackage{xr-hyper}
\newtheorem{rmk}{Remark}[]

\usepackage{algorithm}
\usepackage{algpseudocode}
\usepackage{scalerel}
\usepackage{multicol}
\usepackage{setspace}
\newtheorem{definition}{Definition}
\usepackage[noadjust]{cite}

\usepackage{bbm}
\usepackage[normalem]{ulem}
\usepackage{color}
\usepackage{dsfont}
\usepackage{bm}
\usepackage{setspace}
\usetikzlibrary{automata}
\usetikzlibrary{arrows}
\usetikzlibrary{shapes.geometric, arrows}
\usepackage[]{footmisc}
\usetikzlibrary{positioning}
\usetikzlibrary{calc}
\usetikzlibrary{arrows}
\allowdisplaybreaks

\usepackage{mathtools}
\DeclarePairedDelimiter{\ceil}{\lceil}{\rceil}
\begin{document}

	\title{\fontsize{22.8}{27.6}\selectfont Beamwidth Optimization and Resource Partitioning Scheme for Localization Assisted mm-wave Communications}
	\author{{$^*$}Gourab Ghatak$^{1}$, {$^*$}Remun Koirala$^{2,3,4}$, Antonio De Domenico$^2$, Beno\^{i}t Denis$^2$, Davide Dardari$^4$, Bernard Uguen$^3$, and Marceau Coupechoux$^5$
     \\ \small{ $^1$Indraprastha Institute of Information Technology Delhi (IIIT D), India; $^2$CEA-LETI, Grenoble,
France; $^3$University of Rennes 1-IETR (CNRS UMR 6164), Rennes, France; $^4$DEI, University of Bologna, Cesena, Italy; $^5$LTCI, Telecom Paris, Institut Polytechnique de Paris, France;}}
		\maketitle
 \let\thefootnote\relax\footnote{$^*$These authors have equal contributions in the work. \\ This work has been carried out in the frame of the SECREDAS project, which is partly funded by the European Commission (H2020 EU.2.1.1.7 ECSEL – GA 783119). The work of M. Coupechoux has been carried out at LINCS (http://www.lincs.fr).}
        \vspace{-2.5cm}
	\begin{abstract}
	\vspace{-0.5cm}
We study a \ac{mm-wave} wireless network deployed along the roads of an urban area, to support localization and communication services simultaneously for outdoor mobile users.
In this network, we propose a \ac{mm-wave} initial beam-selection scheme {based on localization-bounds}, which greatly reduces the initial access delay as compared to traditional initial access schemes for standalone \ac{mm-wave} \ac{BS}.
%
%Specifically, to provide initial-access to the users, the base station (BS) tunes the beamwidth of its directional antenna in an iterative manner to precisely locate the user, given tolerable beam-selection error constraints. 
%
%To provide the initial-access to the users, the \ac{mm-wave} \ac{BS} adaptively tunes its beamwidth based on the \sout{iteratively refined} \textcolor{red}{estimated} position of the user, until it reaches a tolerable level of beam-selection error.
%
Then, we introduce a downlink transmission protocol, in which the radio frames are partitioned into three phases, {namely}, initial access, data, and localization, respectively. %The partitioning is controlled by a factor $\beta$, which is adjusted to optimize the localization and communication performances. 
%
%To mathematically analyze the performance of this scheme, we first characterize the beam-selection error and the beam-misalignment error in the localization phase with respect to $\beta$. 
%
%Based on this, we derive the downlink \ac{SINR} coverage probability and  rate coverage probability.
%
We establish a trade-off between the localization and communication performance of \ac{mm-wave} systems, and show how enhanced localization can actually improve the data-communication performance.
Our results suggest that {dense} \ac{BS} deployments enable to allocate more resources to the data phase while still maintaining appreciable localization performance. 
Furthermore, for the case of sparse deployments and large beam dictionary size (i.e., with thinner beams), more resources must be allotted to the localization phase for optimizing the rate coverage.
%Finally we optimize the parameter $\beta$ and the beamwidth of the \ac{mm-wave} antennas, so as to maximize the rate coverage probability of the users.
Based on our results, we provide several system design insights and dimensioning rules for the network operators that will deploy the first generation of \ac{mm-wave} \acp{BS}.
	\end{abstract}
	\vspace*{-1.2cm}

\begin{IEEEkeywords}
\vspace*{-.4cm}
{Millimeter-wave communications, Stochastic Geometry, Localization, Positioning, Initial Access}\end{IEEEkeywords}

\vspace*{-.4cm}

\section{Introduction}
{One of the major catalysts for the development of future mobile networks is the increasing demand for high data-rates.}
Accordingly, due to the large bandwidth available between the 24 GHz and 86 GHz frequency range, \ac{mm-wave} communication is an integral part of the \ac{5G} mobile communication systems~\cite{DeDomenico2017}. 
However, transmissions using high frequencies suffer from large attenuation and sensitivity to blockages~\cite{38.900}.
Nevertheless, owing to the short wavelength of {mm-wave} transmissions, antennas are smaller than those in sub-6~GHz bands. This enables the integration of a larger number of antennas on both transmitter and receiver sides~\cite{larsson2013massive}. Consequently, the higher attenuation can be mitigated using directional beamforming techniques~\cite{rappaport2013millimeter}. {Moreover, the large number of antennas enables the BS to adapt the beam size in order to optimally trade-off the beamforming gain and the coverage area.}
Additionally, the usage of directional beamforming greatly reduces co-channel interference~\cite{ghosh2014millimeter}, which further increases the data-rates.

On the downside, very thin beams raise additional constraints in terms of initial access and coverage~\cite{ghatak2017coverage}, which are challenging in the case of standalone deployment of \ac{mm-wave} \ac{BS}s~\cite{qi2016coordinated}. 
One solution to these challenges in case of a multi-{RAT} network {is to allow} the users to simultaneously receive signals in the {mm-wave} and in the sub-6GHz band, and to use the latter to support the initial access to {the} \ac{mm-wave} systems~\cite{ghatak2017coverage}.
Another approach is first to employ positioning algorithms in order to localize the users with respect to the \acp{BS}, and then, to select the proper \ac{mm-wave} beam to initiate the data transmission~\cite{garcia2016location}. 
In other words, the position and orientation information of the users relative to the \ac{BS} can be used as a proxy for channel information to facilitate beamforming. 
This alleviates the need to undergo an elaborate beam training procedure, which is particularly critical in case of low-latency, high-throughput applications. 
However, this comes at the cost of an increased error in the configuration of transmit and receive beams due to possible inaccuracies in localization~\cite{ghatak2018positioning}. Thus, to completely harness the potential benefits offered by \ac{mm-wave} communications, accurate characterization of the spatial configuration (e.g., the relative distance, \ac{AoD}, and \ac{AoA} between the transmitter and the receiver) is necessary. 

We investigate a \ac{mm-wave} network deployed along the roads of an urban area to support localization and communication services simultaneously. Particularly, we study and optimize a resource partitioning scheme to address jointly localization and communication requirements.
%
%Specifically, we characterize the downlink data-rate coverage probability \sout{of the users in} \textcolor{red}{of} the network and prescribe an algorithm for the selection of the optimal resource-partitioning factor and the optimal beam \sout{from the beam dictionary,} for maximizing the data-rate, while accounting for the inaccuracy in position and orientation estimation\sout{ due to the inadequacy of the localization function that supports this selection}.
\vspace*{-0.3cm}
\subsection{Related Work}
\vspace*{-0.2cm}
The feasibility of providing very high data-rates by operating at mm-wave frequencies is now well established in the literature~\cite{rappaport2013millimeter}. 
Bai {\it et al.}~\cite{bai2015coverage} and Di Renzo~\cite{di2015stochastic} have provided the first works on rate analysis of single-tier and multi-tier mm-wave communications, respectively, in random urban networks. Furthermore, Elshaer {\it et al.}~\cite{elshaer2016downlink} have studied mm-wave systems co-existing with traditional sub-6GHz infrastructure.
However, most of these works either do not consider, or do not fully address the challenges of providing initial access to the mm-wave terminals. 
Ghatak {\it et al.}~\cite{ghatak2017coverage} have also studied networks with co-existing mm-wave and sub-6GHz \acp{RAT}, in which the control signals sent in the sub-6GHz band are used to provide initial access to the mm-wave nodes.
However, they have not provided any algorithm for facilitating the initial access procedure.
In this direction of research, Li {\it et al.}~\cite{li2017design} have studied simple initial access protocols involving hybrid directional beamforming and omni-directional transmissions during the cell-search and random access phases.
According to their finding, the best trade-off between initial access delay and average downlink throughput is obtained using wide beams in the \ac{BS} side and beam-sweeping in the user-side.
The major concern with such protocols remains the high access delays in case of a high number of beams, especially for systems serving low-latency applications.
Recently, Yang {\it et al.}~\cite{yang2018fast} have studied an initial access scheme that substantially reduces the latency with respect to the classical exhaustive and iterative search algorithms.

In the context of user localization, the potential benefits of high-accuracy localization using \ac{mm-wave} beamforming was relatively unexplored until recently~\cite{lemic2016localization}. %In fact, mm-wave beamforming intrinsically allows for both accurate localization and orientation of users with respect to the \ac{BS}~\cite{Destino2017on}.
%
% Lemic {\it et al.}~\cite{lemic2016localization} have shown that localization using mm-wave frequencies is efficient in terms of accuracy, even in the presence of a limited number of anchor nodes.
%
{The initial works in {mm-wave} localization studied how to derive the {CRLB} of the location dependent variables (\textit{e.g.} distance, {AoD}, and {AoA}) considering single \cite{Shahmansoori15} and multiple carriers \cite{Shahmansoori17} using both single \cite{Shahmansoori15} and multipath models~\cite{Shahmansoori17}}.
{Then, the focus was on localization oriented beamforming, considering these theoretical performance bounds.}
{In particular, the authors in \cite{koirala17localization, Koirala18Localization} studied the localization optimal beamforming problem, considering the joint optimization of the {CRLB} of the localization variables for both single and multiple user cases.}
{More recently, in \cite{Kakkavas19}, the authors presented a beamforming strategy to minimize the localization error expressed in the form of the squared position error bound (SPEB).}
From the perspective of joint localization and communication functionalities, Destino {\it et al.} ~\cite{Destino2017on} {and Kumar {\it et al.}\cite{Kumar18On}} have studied the trade-off between communication rate and localization quality in a single user and multi-user mm-wave link respectively.
Typically, the localization performance is characterized by theoretical bounds that model its accuracy~\cite{van2004detection}, whereas, in order to characterize the communication performance, metrics such as user throughput are derived~\cite{li2017design}. 
{Likewise, in \cite{Maschietti17Robust}, the authors present a beam alignment optimization scheme between the transmitter and the receiver considering erroneous position estimations at both ends and scatterers.
In this work, the authors describe a 2-step beam alignment algorithm, firstly at the transmitter independently and then at the receiver following the transmitter's decision.
Recently in \cite{Igbafe19Location}, the authors presented a beam alignment method where, under similar conditions as \cite{Maschietti17Robust}, the transmitter and receiver select the beams in a joint manner, thus outperforming the 2-step method.
Likewise, in \cite{Garcia18Transmitter}, the authors present an iterative localization based beam selection algorithm where the transmitter, in each iteration, selects a refined finer beam based on position and orientation estimation.
The refined beam again improves the estimation and the process continues in a virtuous loop. 
Extending this idea, in \cite{Garcia19Fast}, the authors present the beam selection algorithm at both transmitter and receiving ends.
}
A work that analyses the trade-off between localization and data communication in random wireless networks appeared as a {preliminary} conference version of this paper~\cite{ghatak2018positioning, koirala2018throughput}. 
There, we have used stochastic geometry to derive the \ac{SINR} coverage probability and characterize the data-rate performance during the data service phase given some localization performance during the localization phase. This approach only partially  captures the intricate relation between localization and communication performance, precisely because enhancing localization may improve the downlink data-rate in \ac{mm-wave} systems by reducing beam-selection and misalignment errors.

In this paper, we introduce the notion of effective {data-rate} that not only takes into account the effect of localization error on the downlink data-rate, but also integrates the overhead due to the initial access scheme. {The application of stochastic geometry enables us to formulate the beamwidth optimization and resource-partitioning problem from the perspective of a random user in the network sampled from the distribution of the users. 
%The derived effective rate coverage probability of the typical user thus represents, in an ergodic sense, the fraction of users that achieves a given target rate in bits per second. 
Consequently, the prescribed scheme for optimizing the effective rate-coverage probability of the typical user gives the best {\it expected} beamwidth values and the {\it expected} resource-partitioning parameter for the entire network. This enables us to derive essential system design insights for this network.} The overall contributions are summarized as follows. 
\vspace*{-0.3cm}
\subsection{Contributions and Organization}
\vspace*{-0.2cm}
In this paper, we study a mm-wave network with simultaneous localization and communication services in a one-dimensional (e.g., along roads) scenario. 
%{In a real world deployment of mm-wave BSs along roads, the location of the users from the BSs can vary from a few meters to several kilometers. The stochastic geometry analysis developed for the localization and the data-communication phase gives the network operator an average view of the network performance for both positioning and data-services.}
Particularly, we design and study a downlink transmission scheme where the radio frames are partitioned into initial access, data, and localization phases.
In this paper,
\begin{enumerate}
    % \item We recall the characterize the Bayesian information matrix (BIM) and consequently the \ac{BCRLB} for the joint estimation of the distance, \ac{AoD}, and \ac{AoA} of the \ac{BS}-user link. In contrast to the existing studies, we exploit the distribution of the relative \ac{BS}-user positions given by the stochastic geometry framework, to provide a more accurate \ac{BCRLB} characterization of these parameters. Based on these, we define and derive two performance metrics to evaluate the positioning performance of the system, namely, the beam selection error and the misalignment error.
     
      \item { We recall the {FIM} and consequently the {CRLB} for the joint estimation of the distance and {AoA} of the {BS}-user link. Based on these, we define and derive two {new} performance metrics to evaluate the localization performance of the system, namely, the beam selection error and the misalignment error, which are respectively induced by user's position and orientation estimation errors.}
    %  \item We recall the \ac{FIM} and consequently the \ac{CRLB} for the joint estimation of the distance and \ac{AoA} of the \ac{BS}-user link. Based on these, we define and derive two performance metrics to evaluate the localization performance of the system, namely, the beam selection error and the misalignment error, which are respectively induced by user's position and orientation estimation errors.
    \item With the help of the formulated beam selection and misalignment errors, we design the Tx-Rx best beam pair selection strategy for establishing the initial connection between the \ac{mm-wave} \ac{BS} and the user, which reduces the overhead of beam training significantly as compared to the popular beam sweeping methods~\cite{li2017design}.
    \item  For the data phase, we provide a more accurate characterization of the downlink \ac{SINR} coverage probability as compared to the existing studies~\cite{ghatak2018positioning, koirala2018throughput}, by taking the errors during the localization phase into account. Leveraging on this mathematical characterization, we highlight the non-trivial localization and data-rate trade-off in this system. As an example, it may be intuitive to expect that allocating larger amount of resources to the communication phase increases the data-rate. However, as this results in a shorter localization phase, which leads to less accurate localization of the users, it adversely affects the data-rates. In this work, we optimize this resource partitioning factor (\textit{i.e.}, adapting the resource split between data and localization phases) jointly with the beamwidth of the \ac{mm-wave} \acp{BS} to simultaneously address the localization and communication requirements.
\end{enumerate}
The framework developed in this paper can be used in frame design of an urban mm-wave system, not only for maximization of downlink data-rate, but also to address given data-rate requirements under localization constraints and vice-versa. 
%{We rely on tools from stochastic geometry to give an overview of the system performance across different locations of the network, or across such several networks deployed on different cities.}

The rest of the paper is organized as follows. In Section~\ref{sec:System Model}, we introduce our system model and outline our optimization objectives. Then, in Section~\ref{sec:IA}, we describe the proposed initial access scheme. We provide the performance analysis of the localization and data phases in Section~\ref{sec:LandD}, whereas we discuss the overall system performance in Section~\ref{sec:NRD}. Finally, the paper concludes in Section~\ref{sec:Con}. Table \ref{tab:Param} summarizes the main notations used in this paper.
\vspace{-0.5cm}
\section{System Model}
	\label{sec:System Model}
%	\subsection{Two-Tier Network Model}
\label{sec: Network Model}
    \begin{table*}[!t]
    \small
	\centering
    \caption{Main Notations And System Parameters.} % title of Table
\begin{tabular}{|c  c  c| }
	\hline  Notation& Parameter& Value\\
\hline
    \hline $\xi$, $\lambda$ & BS process and its intensity & $\lambda$ = 5-200 per km. \\  
	\hline {$P_t$}& {Transmitted power from \ac{BS}} &  {30 dBm}\\ 
	\hline  {$\theta_{B/U}$} & {Width of the beams of \ac{BS}/user} & {-}\\
	\hline  {$\theta_{k}$} & {Instance of $\theta_B$} & {-}\\
	\hline  {$\psi/\phi$} & {\ac{AoA}/\ac{AoD}} & {-}\\
    \hline  {$\tau$} & {Delay corresponding to LOS distance between BS and the user} & {-}\\
    \hline  {$\beta$} & {Resource partitioning factor} & {-}\\
	\hline  {$\alpha_{L}$, $\alpha_{N}$}&  {Approximated LOS/NLOS path-loss exponents} &  {2, 4}\\
    %\hline  $\alpha_{SNm}$& NLOS BS path-loss exponent in mm-wave & 8\\
    % \hline  $G_{{ B}}$/$G_{{ U}}$ & mm-wave \ac{BS}/user antenna gain& - \\
   % \hline  $G$ or $G(\theta)$ & Tx/Rx main-lobe directivity gain& - \\
    %\hline  $g$ & Tx/Rx side-lobe directivity gain& - \\
    \hline  ${N_0}$ & Noise power density of the received signal & -174 dBm/Hz\\
    % \hline  $N_0$ & Measurement noise power & -30 dBW\\
    \hline  {$B$} &  {Bandwidth} &  {1 GHz}\\
    \hline  $\sigma_{d}^2/\sigma_{\psi}^2$& Variance of the distance/\ac{AoA} estimation error & -\\
   % \hline  $\sigma_{\psi}^2$& Variance of \ac{AoA} estimation error & -\\
    \hline  {$d_{S}$} &  {BS LOS ball radius} &  {20 m}\\
%    \hline  {$N$} & {Maximum number of beams in the beamforming dictionary} & {-}\\
    \hline 
	\hline 
\end{tabular}
\label{tab:Param}
\end{table*}

\newcommand\scalemath[2]{\scalebox{#1}{\mbox{\ensuremath{\displaystyle #2}}}}

\begin{figure}
\centering
\includegraphics[width=.8\linewidth]{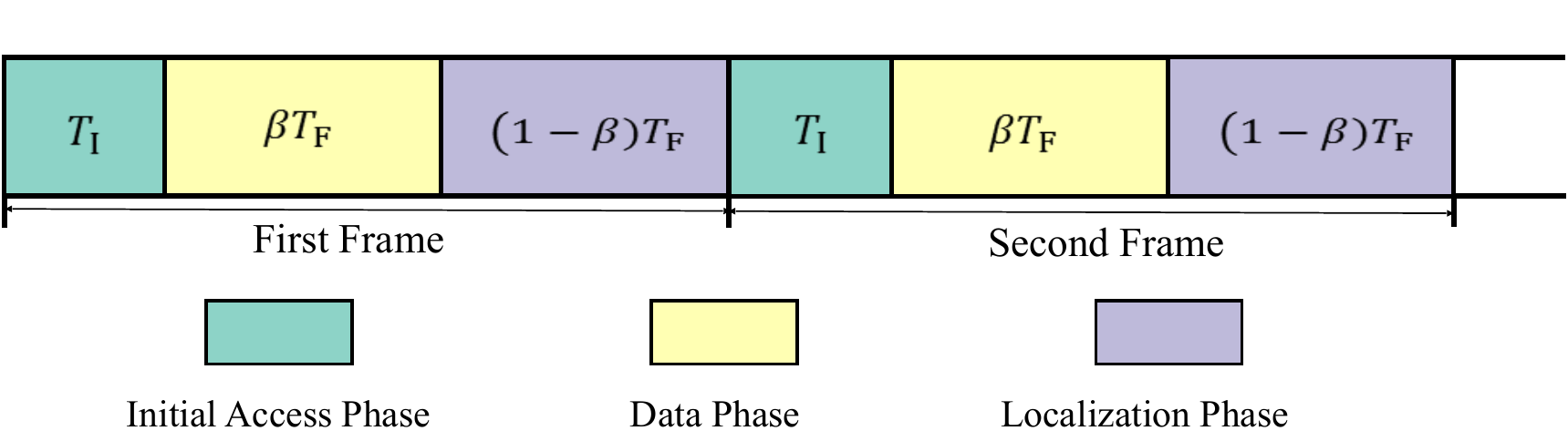}
\caption{The proposed radio frame structure for localization assisted mm-wave communications.}
\label{fig:frame}
\vspace*{-1cm} \end{figure}

We consider a small cell network where multi-RAT \acp{BS} are deployed along the roads to provide high speed data-access to the mobile users by jointly exploiting sub-6GHz and mm-wave bands. 
In this context, we propose a radio {frame structure for} joint communication and localization services, which is illustrated in Fig. \ref{fig:frame}.
Each frame consists of an initial access phase of length $T_{{ I}}$ and a service phase of length $T_{{ F}}$. 
{The access phase enables to establish reliable mm-wave services to the new \acp{UE} that arrive in the system before the start of that frame. This phase is relevant only for the new \acp{UE} and is not repeated for all the \acp{UE}.}
To do so, our approach iteratively increases the resolution of the estimation of the distance and the orientation of the user with respect to the serving \ac{BS}.
The service phase is further partitioned by a factor $\beta$ into a data phase of length $\beta T_{{ F}}$ and a localization phase of length $(1-\beta)T_{{ F}}$.
{{In this paper, we assume that the resources in the initial access and localization phases are perfectly multiplexed across active users, i.e., interference does not affect the localization performance.}}
The service phase is further partitioned by a factor $\beta$ into a data phase of length $\beta T_{{ F}}$ and a localization phase of length $(1-\beta)T_{F}$. The access phase enables to provide reliable mm-wave services to the new \acp{UE} in the system. Accordingly, in this phase, the initial beams at the \ac{BS} and \ac{UE} sides are refined in an iterative manner, until the localization information (ranging and \ac{AoA}$^1$\footnote{$^1${For \ac{AoA} estimation, we can choose one of the popular techniques such as Bartlett technique~\cite{bartlett50periodogram}, Capon technique~\cite{Capon69High} or a subspace based techniques~\cite{schmidt1982signal} (Multiple Signal Classification). For distance estimation, as a basic option, one can simply use \ac{RSSI} based estimation in the \ac{mm-wave} band.}}) reaches a predefined resolution. Then, the data and localization phases follow as depicted in Fig.~\ref{fig:frame}. {Both downlink and uplink are included in our radio frame structure. Specifically, in the initial access and localization phases there is an exchange of downlink and uplink signals to enable precise estimation of the localization parameters. However, in the following, we {analyse} the performance of the data-phase exclusively during downlink communications.} Initially, the \ac{BS} selects the transmit beamwidth $(\theta^*)$ to maximize the effective data rate, which takes into account the localization errors as well, and satisfy the localization service requirements. In the following frames, $\theta^*$ is further adapted to the obtained position and orientation information in order to improve the system performance. {Thus, in the localization phase, the location information of the users are updated. For static users, this information is improved at each subsequent frame. For mobile users, the aim of this phase is to keep a track of the current location so as to facilitate beam-switching if needed.} %{\sout{The optimal transmit beamwidth $(\theta^*)$, which maximizes the data-rate of the \acp{UE} is calculated offline before the service commencement. The characterization of the data-rate takes into account the positioning errors as well. In the initial access phase, the \ac{BS}-\ac{UE} pairs get refined in an iterative manner, until the positioning information of a \ac{UE} is obtained by its serving \ac{BS} with a certain pre-defined resolution. This resolution of positioning is a parameter specified by the device manufacturer that facilitates feasible \ac{mm-wave} service. Once the initial access phase is concluded, the service phase commences with the optimal beamwidth $\theta^*$. In the subsequent frames, the \ac{BS} refines the location information of the \ac{UE} using the service beam (of width $\theta^*$). Thus,}
It is important to note that the beamwidth at the \ac{BS} side that facilitates mm-wave service (in the initial access phase) is different from $\theta^*$. The former is refined in an iterative manner in the initial access phase using the algorithm defined in Section~\ref{sec:IA} to provide access to new \acp{UE}; whereas, the latter is obtained using the framework developed in Section~\ref{sec:LandD}.

%In this section, first we describe the network {geometry} and the propagation model. Then in the subsequent sections, we will characterize the system performance in each of the constituent part of the radio frame. 

\vspace*{-0.5cm}
\subsection{Network Geometry}
\label{sec:NG}
\begin{figure}
    \centering
    \includegraphics[height = 4cm, width = .6 \linewidth]{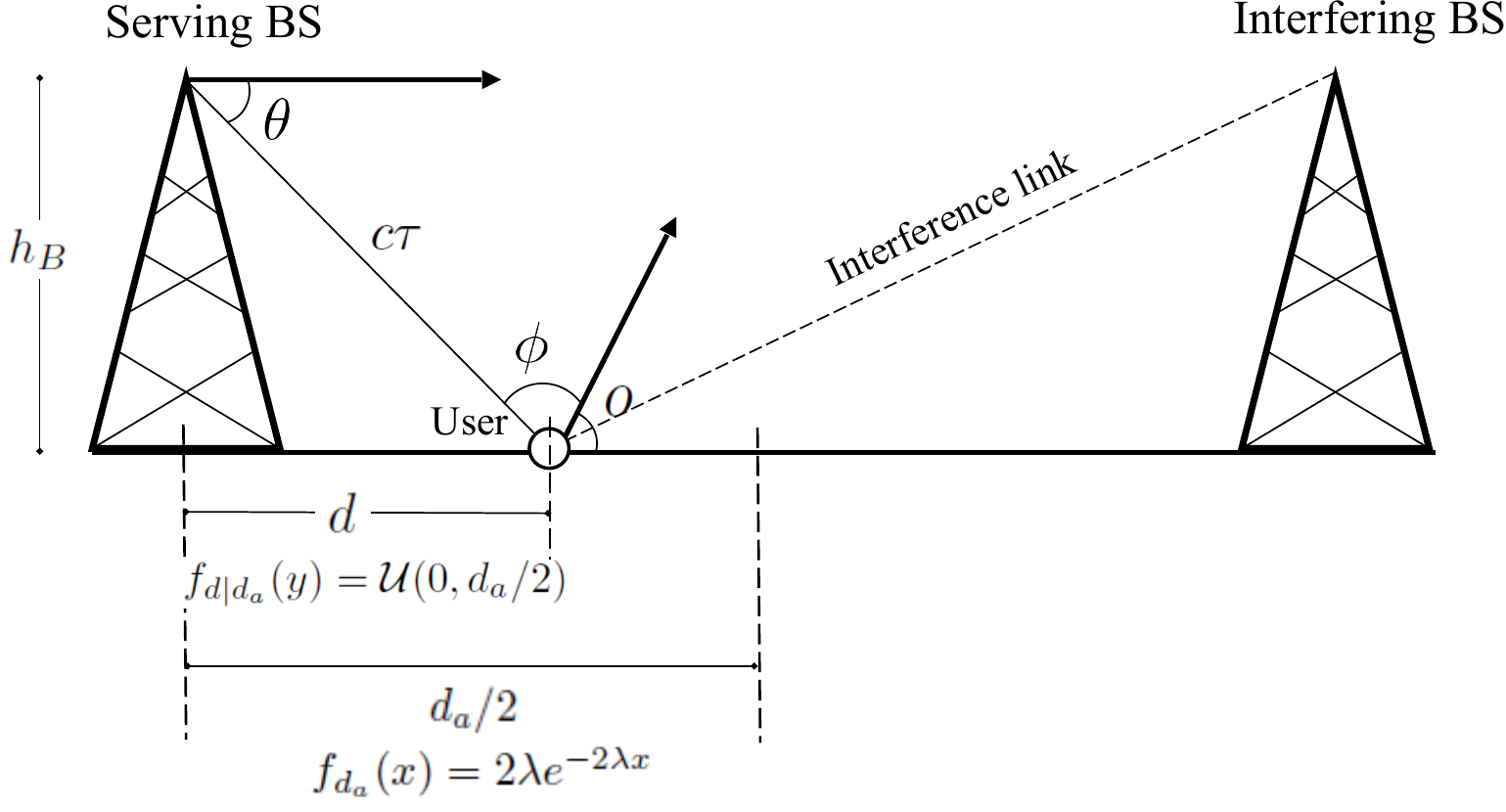}
        \caption{{System model consisting of a serving BS, an interfering BS and a user node at distance $d$ from the serving BS. The figure illustrates the relationship between the \ac{BS} and user positions and the localization variables (distance $d$, \ac{AoD} $\phi$, \ac{AoA} $\psi$ and the user orientation $o$).}} 
    \label{fig:system1}
\vspace*{-1cm} \end{figure}

Let us consider an urban scenario with multi-storied buildings resulting in a dense blocking environment. 
The \acp{BS} deployed along the roads of the city are assumed to be of height $h_{{B}}$ and having a transmit power of $P_{{t}}$. 
Their positions in each street are modeled as points of a one-dimensional \ac{PPP} $\xi$, with intensity $\lambda$ [m$^{-1}$].  {The 1D model assumed in this paper is relevant for the case where the cellular deployment is envisioned to be along roads. As an example, Verizon and AT\&T have both announced plans to deploy 5G infrastructure on lampposts for mobile access~\cite{press3}, respectively. In such scenario, the 1D model assumed in our paper can be utilized by a network operator to derive system design insights and to further fine-tune the deployment parameters.}

%{However, it must also be noted that the methodology developed in this paper can easily be extended to a more general 2D network model, in which case, the coverage of the BSs will be modeled as 2D areas instead of 1D lengths. Accordingly, for the case of 2D networks, the notions developed in this paper e.g., beam selection and misalignment errors, and SINR coverage can be extended by using two variables for representing the locations of the users and solid angles to characterize the beam coverage of the small cells.}

{The users are assumed to be static and located uniformly on the roads with a density $\lambda_{{ U}}$ [m$^{-1}$].
In this regard, it is important to highlight that mobility does not have a large impact on our protocol and performance evaluation methodology. As an example, let us assume vehicular users moving at a speed of 30 km per hour. With 1 ms of frame length, the distance covered by the  user in-between frames is approximately 8 cm, which is considerably small with respect to the coverage area of any beam in the dictionary. When the user speed is very high or the frame length is large, the effect of the mobility on the user localization can be addressed by increasing the variance of the noise in the estimation of the localization variables.}
%In future works, we plan to address both the relation between the noise variance of localization variables and the user speed, and beyond, the possibility of combining our proposed protocol with user tracking algorithms, thus covering even more explicitly dynamic scenarios.}
%

Without the loss of generality, we perform our analysis from the perspective of a \ac{BS} located at the origin and an associated user located at a distance $d$ from the \ac{BS} as illustrated in Fig. \ref{fig:system1}. 
%
%{The user is also assumed to have an orientation $o$ with respect to the reference x-axis.}
%
{The user selects the serving \ac{BS} following a \ac{RSSI} based association.
For a BS located at the origin, the distance from the nearest neighbor (i.e, the closest BS) is given by:
\begin{align}\label{eq:IS-distribution2}
    f_{d_a} (x) = 2\lambda\exp(-2\lambda x). 
\end{align}
Then, assuming that all the BSs have equal transmit power, the coverage area of the BS located at the origin is given by $\frac{d_a}{2}$ on either side of it, where $d_a$ follows the distribution \eqref{eq:IS-distribution2}. Therefore, inside the coverage region of this BS, the location of a random user is uniformly distributed.
Accordingly, the joint probability distribution of the distance $d$ and the coverage area $d_a$ is given by $f_{d_a,d}(x,y) = f_d(y|d_a = x)f_{d_a}(x)$~\cite{chiu2013stochastic}}, where %\textcolor{red}{[To be corrected]}
\begin{align}
f_d(y|d_a=x) = \begin{cases}
x^{-1}; \quad 0 \leq y \leq x \\
0; \quad \text{otherwise}
\end{cases}.
\label{eq:dist}
\end{align}
%In other words, the user is uniformly located in the \ac{BS} coverage area, which in turn is governed by the intensity of deployment $\lambda$. 
Thus, each \ac{BS} is associated with a service coverage area of length $d_{{ a}}$, which is distributed as~\eqref{eq:IS-distribution2}.

{In the following, we denote the user orientation with respect to the reference x-axis as $o$, the \ac{AoA} at the user as $\psi$ and the \ac{AoD} at the \ac{BS} as $\phi$. }
{As depicted in Fig. \ref{fig:system1}, the relation between the position of the BS and the user with the delay $\tau$, \ac{AoD}, \ac{AoA} and the user orientation are:}
%{\begin{subequations}\nonumber
    \begin{align}
            \tau &=\sqrt{d^2+h_B^2}/c,\quad
            \phi &= \cos^{-1} \left({d}/{ \sqrt{d^2+h_B^2}}\right), \quad
            \psi &= \pi - \cos^{-1} \left({d}/{\sqrt{d^2+h_B^2}}\right)-o, \nonumber
            \vspace{-1cm}
        \end{align}
%\end{subequations}
where $c$ is the speed of light.
It must be noted that in our 1D scenario, $\phi$ is dependent directly on $d$.
We also assume that the orientation of the users are unknown, and accordingly, we consider that the distribution of the {initial \ac{AoA} of the user} $f(\psi)$ is uniform between $0$ and $2 \pi$. 

\vspace*{-0.5cm}
\subsection{Millimeter-Wave Beamforming}
Our analysis consists of two parts, the first one involving derivation of \ac{CRLB} for the localization phase and, then, the derivation of the user performance in the data phase. 
For the derivation of \ac{CRLB} for \ac{AoA} estimation, the angular information is derived from the antenna array response, hence we use the {\it \ac{ULA}} model
\cite{Shahmansoori17} with an antenna spacing of {half the carrier wavelength}.
{On the other hand}, in order to simplify the analysis of the data phase, we approximate the beamforming by a \textit{sectorized model}~\cite{hunter2008transmission}, where the transmitted and received beams are divided into two sectors, a {\it main lobe} sector whose antenna gain depends on the beamwidth $\theta$ and a side lobe sector with a fixed gain. {Here, the term main lobe stands for the angular region of the antenna pattern centered around the axis of maximum gain and aperture equal to the half-power beamwidth of the pattern. } {We assume that the BSs do not cater to multiple users or transmit multistream data, simultaneously.} {Accordingly, we assume the existence of a single RF chain with analog beamforming.}

{Accordingly,} in the sectorized model, the antenna gain {at the BS side and user side $G_x(\theta_x)$, where $x \in \{\mbox{B}, \mbox{U}\}$, is given by} \cite{Ghadikolaei15}
\begin{equation}
G_x(\theta_x) =\begin{array}{l}\label{eq:Gain} 
\begin{cases}
\gamma_x(\theta_x)=G_0\frac{2\pi-(2\pi-\theta_x)\epsilon}{\theta_x}, &\mbox{in the main lobe},\\
  g=G_0\epsilon, &\mbox{otherwise}, 
  \end{cases} 
    \end{array} 
\end{equation}
{where $G_0$ is the antenna gain of an equivalent omnidirectional beam (i.e., $\theta_x=2\pi$) and $\epsilon$ is a small positive constant $\ll 1$}.
%\textcolor{magenta}{MC: (3) is not beautiful and confusing. First, we don't understand what is the "main lobe", it is not defined. Second, there is a confusion bw $\theta_x$ with x  in \{B,U\} in (3) and the $\theta_i$ in the figure. Third, G should be a function of a generic angle $\theta$ and may depend then on a parameter that describes the main lobe position.}
In the {\it \ac{ULA} antenna model}, each \ac{BS} and user is assumed to be equipped with mm-wave \ac{ULA} directional antennas consisting of {$M_{B}$} and  {$M_{U}$} antenna elements respectively. Then, the \ac{BS} antenna array response is:
\begin{equation}\label{eq:a_BS}
	\bm a_{B}( {\phi}) = \frac{1}{\sqrt{M_{B}}}\left[1, e^{j \frac{2 \pi \kappa f_{c}}{c} \sin (\phi)}, \cdots,  e^{j (M_{B}-1) \frac{2 \pi \kappa f_{c}}{c} \sin (\phi)} \right],
\end{equation}
where $\kappa$ is the inter-element distance in the antenna system and $f_{{ c}}$ is the center frequency of the mm-wave system. The user antenna response $\bm a_{{ U}}(\psi)$ is simply obtained by replacing $\phi$ with $\psi$ and {$M_{{ B}}$} with {$M_{{ U}}$} in \eqref{eq:a_BS}.

{Let $\bm w{(\theta_{{ B}})} \in \mathbb{C}^{M_B}$ and $\bm w{(\theta_{{ U}})} \in \mathbb{C}^{M_U}$ represent the transmit and receive beamforming vectors. As defined in \cite{yang2018fast}, the width of the beam can be controlled by changing the number of elements $M_B$ and $M_U$ in the antenna array.
%\textcolor{red}{define these vectors as function of the beamwidth as in 
%Fast and Reliable Initial Access with Random Beamforming for mmWave Networks}
Then, the beamforming gains for the BS and the user are given by {$G_{{ B}}(\theta_{{ B}}) = |\bm a_{{ B}}^{H}({\phi}) \bm w{(\theta_{{ B}})}|^2$} and {$G_{{ U}}(\theta_{{ U}}) = |\bm a_{{ U}}^{H}({\psi}) \bm w{(\theta_{{ U}})}|^2$}, respectively.} %\textcolor{red}{I would rather define the antenna gains as function of the precoding and combining matrices instead to save some space later in the document}} % {[BD: $\bm f{(.)}$ (resp. $\bm w{(.)}$) should be a function of the AoD (resp. AoA), but not only of the beamwidth $\theta_{{ B}}$ (resp. $\theta_{{ U}}$), right ? In addition, $\bm a_{{ U}}^{H}(.)$ should be a function of UE orientation, rather than $d$ ? Please double-check and correct if needed...]}
\vspace{-0.5cm}
\subsection{Beam Dictionary}
\begin{figure}
\centering
\subfloat[]
{\includegraphics[height = 3cm, width=7cm]{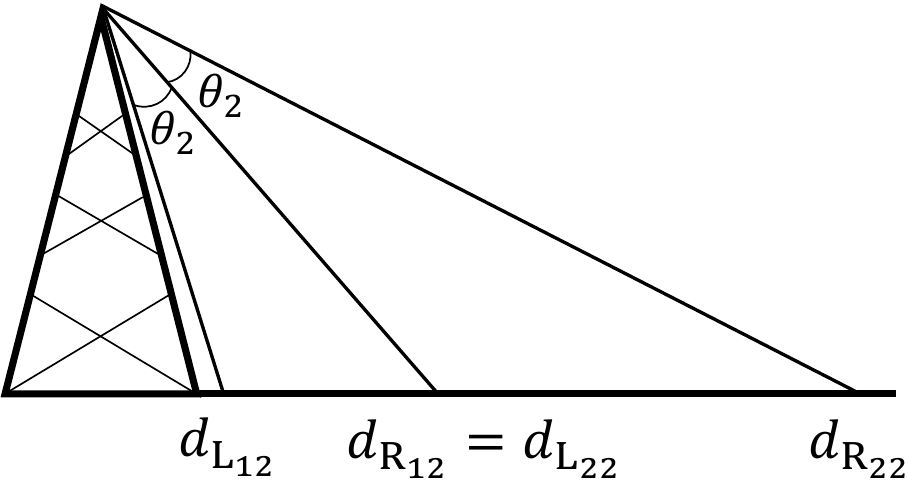}
\label{fig:Sys1}}
\hfil
\subfloat[]
{\includegraphics[height = 3cm, width=7cm]{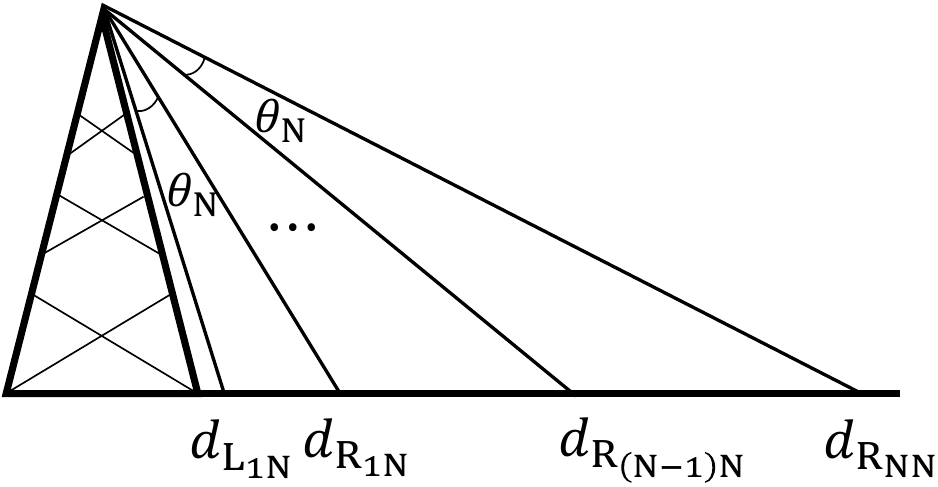}
\label{fig:Sys2}}
\caption{Illustration of the beam dictionary elements in case of (a) $2$ beams and (b) $N$ beams.}
\label{fig:Sys}
\vspace*{-1cm} \end{figure}
We assume that each \ac{BS} designs a sub-6GHz assisted mm-wave beamforming database.
% 
%The user, depending on the coverage area it lies in, associates with the corresponding \ac{BS} \textcolor{red}{This sentence is disconnected from the text here}.
%\footnote{$^{{1}}$In a more simplistic scenario, when the \acp{BS} are assumed to be equi-spaced (e.g., \ac{BS} deployment on top of lamp-posts~\cite{cudak2014experimental}), the coverage area becomes equal for all the \acp{BS}. \textcolor{magenta}{MC: do we need this footnote ? What is the relatioship with the sentence above ?}}
Specifically, each \ac{BS} is capable of having beam dictionaries of different sizes, where each beam dictionary is composed by  {a} set of beams characterized by the same width. 
The size of the beam dictionary denotes the number of beams that characterizes the dictionary.
Furthermore, we assume that the main lobes of different beams of the same dictionary are non-overlapping. Together, the beams of a dictionary provide complete coverage in the geographical coverage area (i.e., the Voronoi cell) of the \ac{BS} as shown in Fig. \ref{fig:Sys}.  %\textcolor{magenta}{MC: on the figure, it seems that you cover only one side of the cell.}
Consequently, the larger the number of beams in the dictionary, the smaller is the beamwidth. It must be noted that for a typical \ac{BS} deployed along the road, the neighbor \acp{BS} on either side may not be located at the same distances from it. As a result, the beam dictionary maintained at the \acp{BS} would contain the cell size information for both the sides of them. Without loss of generality, in what follows, we focus on one side of the typical \ac{BS}.

Let $\theta_1 = \arctan\left(\frac{d_{{ a}}}{h_{{ B}}} \right) - \arctan\left(\frac{d_{{ L}_{11}}}{h_{{ B}}} \right)$ be the beamwidth of the beam that provides total coverage of the area $d_{{ a}}$, where $d_{{ L}_{11}} = d_{{ L}_{12}}= d_{{ L}_{1N}} = 0$ is the starting point of the coverage area (as illustrated in Fig. \ref{fig:Sys}) and $h_{{ B}}$ is the height of the \ac{BS}.
Then, for the beam dictionary size $k$, the beamwidth is defined by $\theta_k = \theta_1/k$.
Now, depending on this beamwidth $\theta_k$ and total number of beams, the left and right boundaries of each main lobe coverage positions of the $j$-th beam ($1 \leq j \leq k$) are {denoted} as $d_{{ L}_{jk}}$ and $d_{{ R}_{jk}}$.
The non-overlapping and adjacent assumption of the beams implies that $d_{{ R}_{jk}} = d_{{ L}_{(j+1)k}}$, $\forall j < k$.

Hence, {we define} the beam dictionary database $\mathcal{DB}$ of a mm-wave \ac{BS} as a lower triangular matrix consisting of all feasible beams for each beam dictionary. 
Each element $\mathcal{DB}_{k,j}$ of $\mathcal{DB}$,  where $j \leq k,$ consists of a tuple $(\theta_k, d_{{ L}_{jk}}, d_{{ R}_{jk}})$ corresponding to the $j$-th beam of the $k$-th beam dictionary. The elements of the tuple indicate respectively a) the width of the beam, b) the left boundary, and c) the right boundary of the main lobe of the beam (according to the \textit{sectorized model}), as illustrated in \eqref{eq:DB}. 
Then, for $k$-th beam dictionary, the $j$-th beam has a coverage area $\mathcal{C}_{jk}  = d_{{ R}_{jk}} - d_{{ L}_{jk}}$.
The steps for designing the beam dictionary at a mm-wave \ac{BS} are:
\begin{enumerate}
\item After being deployed, the new \ac{BS} exchanges inter-BS signals in the sub-6GHz band to discover its geographical location on the street$^2$\footnote{$^2$ Such prior geo-referencing, anyway required for mapping geographical coverage, can also be performed in alternative ways such as the GPS.}, with respect to its neighbouring mm-wave \ac{BS}s$^3$\footnote{$^3$This information can be provided a-priori by the operator during the deployment phase.}.  
%
 %This allows the \acp{BS} to discover the neighboring \ac{BS} locations. s
%
Using this information, a \ac{BS} maps its own geographical coverage area with respect to its neighbors. As all the \acp{BS} are assumed to have the same transmit power, the cell boundaries are midway between two neighboring \acp{BS} as illustrated in Fig. \ref{fig:system1}.  
\item For each value of beam dictionary $k \in \{1, 2, \ldots,N\}$, the BS calculates the coverage areas of the associated beams as $\mathcal{C}_{j,k} = d_{{ R}_{jk}} - d_{{ L}_{jk}}$, where:
\begin{align}\label{eq:dR}
d_{{ R}_{jk}} &= h_{{ B}} \tan\left(\arctan\left(\frac{d_{{ L}_{jk}}}{h_{{ B}}}\right) + j\theta_k\right), \quad j = 1, 2, \cdots,k, \\ 
d_{{ L}_{jk}} &=\begin{array}{l}\label{eq:dL}
\begin{cases}
d_{{ R}_{(j-1) k}}, \quad &j = 2, \cdots, k,\\
  0,  \quad &j=1. 
  \end{cases} 
    \end{array} 
\end{align}

\item The resulting data-base is thus lower triangular matrix as follows: 
\begin{align}\label{eq:DB}
\scalemath{.9}{\mathcal{DB} = \begin{bmatrix}
(\theta_1, d_{{ L}_{11}}, d_{{ R}_{11}}) & 0 & 0 & \ldots  & 0 \\
(\theta_2, d_{{ L}_{12}}, d_{{ R}_{12}}) & (\theta_2, d_{{ L}_{22}}, d_{{ R}_{22}}) & 0 & \ldots  &0 \\
\vdots & \vdots & \ddots &\ddots &\vdots\\
(\theta_N, d_{{ L}_{1N}}, d_{{ R}{1N}}) & (\theta_N, d_{{ L}_{2N}}, d_{{ R}_{2N}}) & \ldots & \ldots &(\theta_{N}, d_{{ L}_{NN}}, d_{{ R}_{NN}})
\end{bmatrix},}
\end{align}
where the $k$-{th} row consists of the beam dictionary of size $k$ beams and contains the information about the width and the main lobe coverage areas of the corresponding beams. %The values $d_{L_{jk}}$ and $d_{R_{jk}}$ respectively represent the corresponding left and right extremes of the coverage area of the $j$-{th} beam with beamwidth $\theta_k$.
\end{enumerate}

\begin{figure}
\centering
\subfloat[]
{\includegraphics[width=7cm]{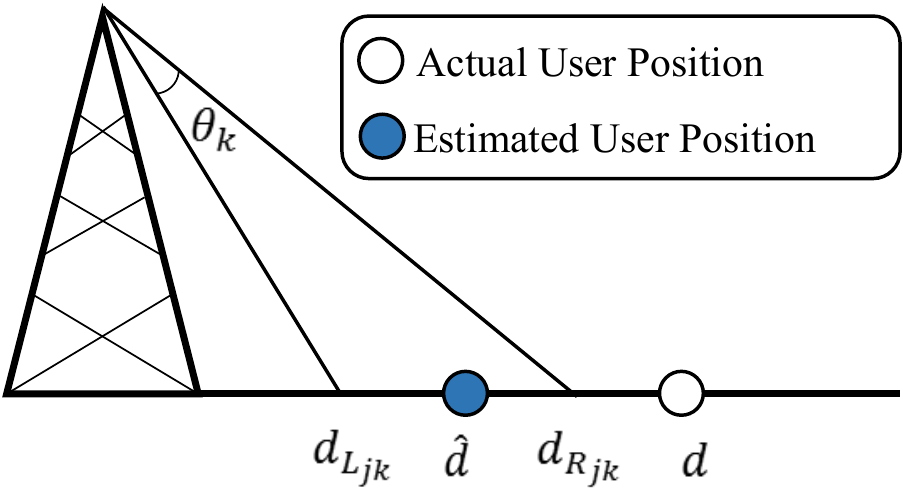}
\label{fig:BS}}
\hfill
\subfloat[]
{\includegraphics[width=7cm]{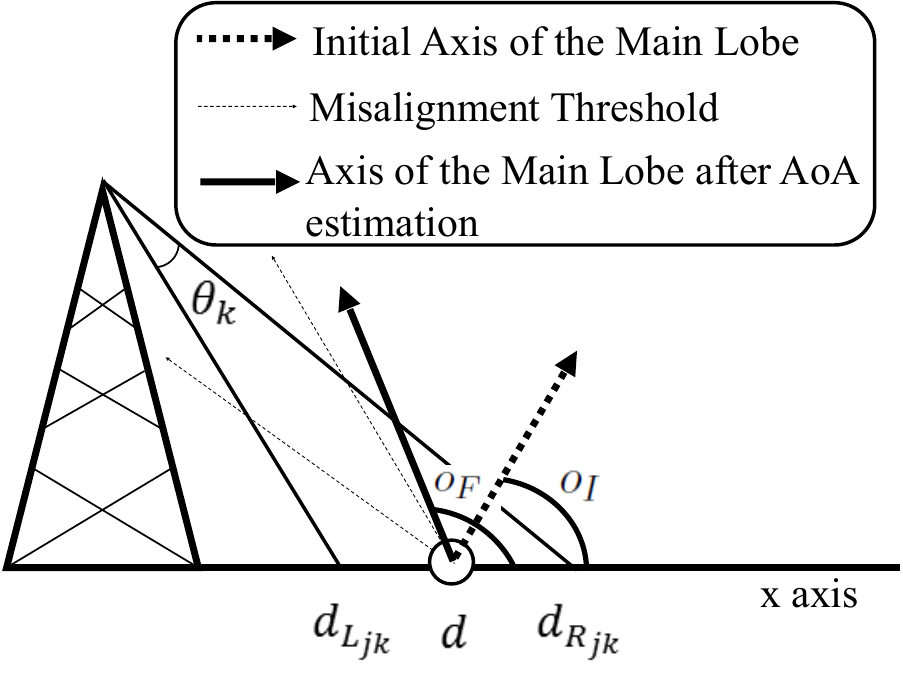}
\label{fig:MS}}
\caption{Illustration of the (a) beam selection error  and (b) misalignment error.}
\vspace*{-1cm} \end{figure}
%\begin{figure}
%\centering
%\includegraphics[width=8cm]{BS.pdf}
%\caption{Example of beam selection error.}
%\label{fig:BS}
%\vspace*{-1cm} \end{figure}
%\begin{figure}
%\centering
%\includegraphics[width=8cm]{MS.pdf}
%\caption{{Example of misalignment error.} The solid arrow represents the axis of the main lobe of the user antenna. The user-BS beam pair is assumed to be misaligned if the axis of the main lobe of the user antenna is outside the orientation thresholds (depicted by dashed arrows).\textcolor{magenta}{MC: this figure is very helpful but $\psi$ should appear somewhere to be complete.}\textcolor{green}{Remun, can you please rename the boresight as "axis of maximum gain" as discussed earlier?}}
%\label{fig:MS}
%\vspace*{-1cm} \end{figure}

Following the description of the beam-dictionary, we define two critical metrics of the system, which we will use to characterize the performance of the localization phase. %\textcolor{red}{ let define here and harmonize between $\mathcal{P}_{{ BS},{j,k}}\left(d,\sigma^2_d\right)$ and $\mathcal{P}_{MA}(x, {\psi})$ }.
\begin{definition}
The beam-selection error is defined as the event that a UE located in $\mathcal{C}_{j,k}$ is estimated to be at $\hat{d}$, outside of $\mathcal{C}_{j,k}$, and accordingly, it is allotted a different beam than $(\theta_j,d_{L_{j,k}},d_{R_{j,k}})$. Let us denote as $\sigma^2_d$ the variance of the distance estimation error; the probability of beam-selection error ($\mathcal{P}_{BS}$), given that the UE is located at a distance $d$, is defined as  
\begin{align}
    \mathcal{P}_{{ BS},{j,k}}\left(d,\sigma^2_d\right) = \mathbb{P}\left(\hat{d}(d,\sigma^2_d) \notin \mathcal{C}_{j,k} | d \in \mathcal{C}_{j,k} \right).
\end{align}
\end{definition}
This event is depicted in Fig.~\ref{fig:BS}.
\begin{definition}
The beam misalignment error is defined as the event that, {after the \ac{AoA} estimation}, the UE beamforms towards a direction such that the axes of the main lobe of the UE and BS antennas have an angular separation greater than a predefined threshold $\nu$. {Let us denote as $o_I$ the initial user orientation and as $\hat{\psi}$ and $\sigma^2_\psi$ the estimated \ac{AoA} and the variance of the \ac{AoA} estimation error, respectively. After the \ac{AoA} estimation, the user orients its main lobe towards the direction of $\hat{\psi}$ in order to align it towards the BS main lobe. The new orientation of the user main lobe is denoted by $o_F$ in Fig.~\ref{fig:MS}.} The probability of misalignment error, given that the UE is located a distance $d$ (i.e., $\mathcal{P}_{MA}$) is then defined as %, averaged over all initial orientations of the UE is defined as
\begin{align}
    \mathcal{P}_{MA,j,k}(d, {\psi}, \sigma^2_\psi) = \mathbb{P}\left(|\psi - \hat{\psi}(d,\sigma^2_\psi)| \geq \nu(\theta_k,\theta_U)\right).
\end{align}
\end{definition}
This event is depicted in Fig.~\ref{fig:MS}.
%{In the following, we denote the user orientation with respect to the reference x-axis as $o$, the \ac{AoA} at the user as $\psi$ and the \ac{AoD} at the \ac{BS} as $\phi$. As depicted in Fig. \ref{fig:system1}, the relation between the position of the BS and the user with the delay ($\tau$), \ac{AoD}, and \ac{AoA} are as follows}
%{\begin{subequations}\nonumber
%    \begin{align}
%            \tau &=\sqrt{d^2+h_B^2}/c;\\
%            \phi &= \cos^{-1} \left({d}/{ \sqrt{d^2+h_B^2}}\right); \\
%            \psi &= \pi - \cos^{-1} \left({d}/{\sqrt{d^2+h_B^2}}\right)-o,
%        \end{align}
%\end{subequations}
%where $c$ is the speed of light.}

% \begin{definition}
% The beam misalignment error is defined as the event that the UE beamforms towards a direction such that axes of the main lobe of the UE and BS antennas have an angular separation greater than a predefined threshold $\nu$. Let us denote as $\hat{\psi}$ and $\sigma^2_\psi$ the estimated \ac{AoA} and the variance of the \ac{AoA} estimation error; the probability of misalignment error, given that the UE is located a distance $d$ (i.e., $\mathcal{P}_{MA}$) is defined as %, averaged over all initial orientations of the UE is defined as
% \begin{align}
%     \mathcal{P}_{MA,j,k}(d, {\psi}, \sigma^2_\psi) = \mathbb{P}\left(|\psi - \hat{\psi}(d,\sigma^2_\psi)| \geq \nu(\theta_k,\theta_U)\right).
% \end{align}
% \end{definition}
% This event is depicted in Fig.~\ref{fig:MS}.

\vspace*{-0.5cm}
\subsection{Blockage, Path-Loss, and Signal Propagation}
Due to the {presence of buildings and other obstacles}, the communication links can either be in \ac{LOS} or \ac{NLOS} state. We assume a \ac{LOS} ball model for {characterizing} the blockage, similar to that in~\cite{bai2015coverage}, with a \ac{LOS} ball radius $d_{{ S}}$. Thus, all the \acp{BS} present within a distance $d_{{ S}}$ from the user are assumed to be in \ac{LOS}, whereas, the \acp{BS} lying beyond $d_{{ S}}$ are assumed to be in NLOS. Accordingly, the \ac{LOS} \ac{BS} process is denoted by $\xi_L$ and the \ac{NLOS} \ac{BS} process is denoted by $\xi_N$.  %\textcolor{red}{{define here $\xi_{{ L}}$ and $\xi_{{ N}}$}; these are the processes related to the small cells in LOS and NLOS; these notations are then used in the demonstrations.}. 
%
%{[BD: Unclear at first reading. You probably mean "The LOS (resp. NLOS) channel status is reference to as $\xi_{{ L}}$ (resp. $\xi_{{ N}}$)?]}.
%
Furthermore, because of the low local scattering in mm-wave communications, we consider a Nakagami fading $f$ with parameters $N_L$ and $N_N$ for the \ac{LOS} and \ac{NLOS} paths, respectively, and variance equal to 1~\cite{7593259}.
Additionally, we assume a path loss model where the power at the receiver located at a distance $d$ from the \ac{BS} is given by $P_{{ r}} = K P_{{ t}}  |f|^2  {G_{{ B}}(\theta_{{ B}})}{G_{{ U}}(\theta_{{ U}})} (d^2 + h_{{ B}}^2)^{\frac{-\alpha}{2}}$, where $K$ is the path loss coefficient, $P_{{ t}}$ is the transmitted power, and $\alpha$ is the path loss exponent. In our model, $\alpha = \alpha_{{ L}}$ or $\alpha_{{ N}}$ depending on whether the link is in \ac{LOS} or \ac{NLOS} state, respectively. %Thus, the average SNR can be written as $\frac{K\cdot P \cdot G_0 \cdot (d^2 + h_B^2)^{\frac{-\alpha}{2}}}{N_0 \cdot B}$. 

Let us assume that the received signal suffers from a zero-mean additive Gaussian noise with two-sided noise power spectral density of ${N_0}$ [dBm/Hz].  As a result the \ac{SINR} in the data-communication phase is given by:
\begin{align}
    \text{SINR}_C = \frac{K P_{{ t}}{|f|^2} G_{{ B}}(\theta_B) G_{{ U}}(\theta_U)\left({h_{{ B}}^2 + d^2}\right)^{-\frac{\alpha}{2}}}{{N_0 B}  + \sum_{i \in \mathcal{I}}K P_{{ t}}{|f_i|^2}g^2\left({h_{{ B}}^2 + d_i^2}\right)^{-\frac{\alpha}{2}}},
\end{align}
where $\mathcal{I}$ refers to the set of interfering \acp{BS}.
%{Here it is important to note that the BSs do not point their main lobe towards the coverage areas of other BSs. Consequently, for any given user, the interference power is sent from the side lobe of the interfering BSs, and never from the main lobe. Similarly, with our orientation estimation procedure, the users are assumed to beamform towards their serving BS. As a result, the gain from the interfering BSs is always $g^2$.}

{Contrary to the data communication phase, in the localization phase, we do not consider the effect of interference. This is primarily because {we assume that} the localization estimation occurs using signals transmitted in the control channel, which is assumed to be interference-free due to the usage of orthogonal resources for transmitting the pilots. This is in line with classical and recent works on mm-wave localization \cite{Shahmansoori17, Garcia17}.}
% Let the power of noise in the localization phase be given by $N_0$ [dBm], where this noise originates from measurement errors during estimation of the delay and the orientation of the \acp{UE}.
{Hence, the \ac{SNR} in the localization process is given as:
\begin{align}
    \text{SNR}_L &=  \frac{K P_{{ t}}{|f|^2} G_{{ B}}(\theta_B) G_{{ U}}(\theta_U)}{{N_0 B}}\left({h_{{ B}}^2 + d^2}\right)^{-\frac{\alpha}{2}}.
    \label{eq:SNR}
\end{align}}

% {Contrarily to the data communication phase, in the localization phase we do not consider the effect of interference. It has been demonstrated in \cite{Abu18Error} that under the favourable conditions of high carrier frequency, large bandwidth and large number of antenna elements, the resolution of estimating delay and angle is large enough that the multipath elements can be resolved path by path and the delay and AoA can be estimated without any biases. Such as assumption holds in the mm-wave scenario.}
% \textcolor{red}{GG: i do not understand why the ability to resolve multiple paths means interference is not present? Should we rather say that we assume that the localization takes place in the control channel which has dedicated carrier etc. devoid of interference?}
% Let the power of noise in the localization phase be given by $N_0$ [dBm], where this noise originates from measurement errors during estimation of the delay and the orientation of the \acp{UE}.
% Hence, the \ac{SNR} in the localization phase is given as :
% \begin{align}
%     \text{SNR}_L &=  \frac{K P_{{ t}}{|f|^2} G_{{ B}}(\theta_B) G_{{ U}}(\theta_U)}{{N_0 B}}\left({h_{{ B}}^2 + d^2}\right)^{-\alpha}.
%     \label{eq:SNR}
% \end{align}
% Moreover, let us assume that the communication in the data phase suffers from a zero-mean additive Gaussian noise with two-sided noise power spectral density of $\sigma_{{ N}}^2$ [dBm/Hz]. As a result the \ac{SINR} in the data-communication phase is given by:

% \textcolor{red}{As said, I would include the SINR here}

\vspace*{-1cm}
\section{Initial Beam-Selection Procedure}
\label{sec:IA}
In this section, we discuss our initial beam-selection procedure for a user arriving in the \ac{mm-wave} network.
In this procedure, the \ac{BS} and the user select appropriate {beam pairs}, $\theta_{{ B}}$ and $\theta_{{ U}}$ respectively, based on the localization accuracy required for the initial access.
%
% In the subsequent sections, we will characterize the system performance and obtain an optimal beamwidth $\theta^*$ to maximize the data-rate or to address given data-rate and localization constraints. At this point, it must be noted that  $\theta_{{ B}}$ for facilitating the user initial beam-selection is not necessarily $\theta^*$ [\textcolor{red}{Not defined}]. 
%
% While the former is iteratively obtained for achieving the positioning accuracy required for \ac{mm-wave} access, the later is optimized offline as per the user \ac{QoS} requirements. The steps for the initial beam-selection are as follows.
%Furthermore, let us assume that the optimized service beamwidth $\theta^* \in \mathcal{DB}$ (which is obtained offline and will be elaborated in the next sections of the paper) is used for rest of the phases in the frame (\textit{i.e.} localization and data service phases).  
%
% It must be noted that the beamwidth that provides sufficiently accurate positioning estimates for facilitating the initial access of an user is not necessarily $\theta^*$. 
% %
% While the former is selected purely on the basis of localization accuracy, the later is optimized offline as per the QoS requirements of the user. The steps for the initial access are as follows.
%
%Here, we assume that both the bands in the BS are capable of transmitting C-plane data, as opposed to the U-plane data, which is transmitted only in the mm-wave band. 
%

\tikzstyle{block} = [rectangle, draw,  
    text width=20em, text centered, minimum height = 8em]
\tikzstyle{line} = [draw, -latex', width=0.25mm]

% \begin{figure}
% \centering
%   \begin{tikzpicture}[node distance = 2cm, auto]
%     % Place nodes
%     \node [block, line width=1pt, align=left] (init) {1. Estimate $\hat{d}$ with variance $\sigma_d^2$.\\ 
%     2. Choose $\theta_{jk}$ such that: 
%   \begin{align}
%       &k = \argmax_{k} \left[\mathcal{P}_{BS}(\theta_k, \sigma_d) \leq  \delta_{BS}\right] \nonumber \\
%       \text{and } & L_{jk} \leq \hat{d} \leq R_{jk} \nonumber
%   \end{align}
% };
%     \node [block, below of=init, node distance=5cm, , line width=1pt, align=left] (identify) {1. Estimate $\hat{\psi}$ with variance $\sigma_\psi^2$.\\ 
%     2. Choose $\theta_{U}$ such that: 
%   \begin{align}
%       &\theta_U = \argmin_{\theta} \left[\mathcal{P}_{MA}(\theta_k, \sigma_\psi) \leq  \delta_{MA}\right] \nonumber \\
%       \text{and } & 0 \leq \theta_U \leq \frac{\pi}{2} \nonumber
%   \end{align}
% };

%     % Draw edges

% % (5,-2.2) edge [bend right][left] {} (1.5,-4);
% \draw [->, line width=1pt] (4.35,0) to (5.35, 0) to (5.35, -5) to (4.35, -5);
% \draw [<-, line width=1pt] (-4.35,0) to (-5.35, 0) to (-5.35, -5) to (-4.35, -5);

% % \draw [->, line width=1pt] (4.32,0) to [bend left=90] (4.32,-5);
% % \draw [->, line width=1pt] (-4.32,-5) to [bend left=90] (-4.32,0) ;
% \end{tikzpicture}
% \caption{The \ac{BS} and user side distance and \ac{AoA} estimation and beamwidth update loop.}
%     \label{fig:my_label}
% \vspace*{-1cm} \end{figure}

\tikzstyle{decision} = [diamond, draw, text width=4.5em, text badly centered, node distance=3cm, inner sep=0pt, line width=1pt]
\tikzstyle{block} = [rectangle, draw, text centered, rounded corners, minimum height=4em, line width=1pt]
\tikzstyle{line} = [draw, -latex', line width=1pt]
\tikzstyle{cloud} = [draw, ellipse, node distance=3cm,
    minimum height=2em, line width=1pt]
\tikzset{io/.style ={trapezium, draw, minimum width=2.5cm,trapezium left angle=60, trapezium right angle=120, line width=1pt}}
    
\begin{figure}
\small{
\includegraphics[width = .86\linewidth]{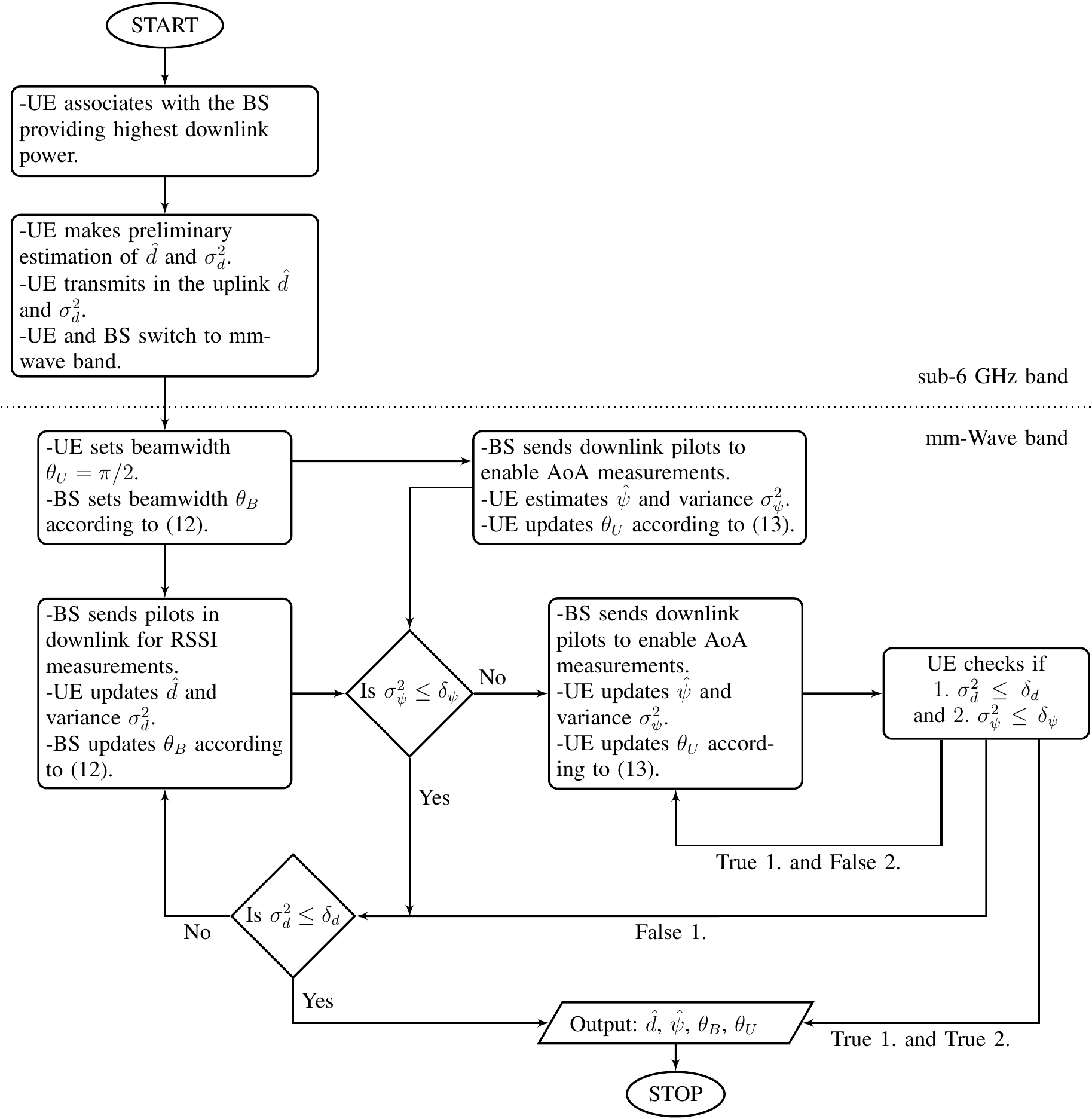}
\caption{Flowchart representing the \ac{BS} and user beam selection procedure. %\textcolor{magenta}{MC: Nice idea but what is ksi? Also the algorithm may loop forever. An idea to avoid having extra parameters $\delta$s that are difficult to set and may be never achieved, is to say: we have a delay budget or in terms of number of iterations and we iterate within this budget. Is it run only at initial access or continuously during data transmission? Ok only for initial access.
}
\label{fig:my_label}
\vspace*{-1cm} 
}
\end{figure}

\begin{enumerate}
% \item When a new user arrives in the network, it associates with the \ac{BS} that provides the highest downlink received power in the sub-6GHz band. The  associated \ac{BS} sends {downlink pilots} in the sub-6GHz band$^{{6}}$\footnote{$^{{6}}$ Without loss of generality, this coarse user localization can be obtained by means of external technologies as well (e.g. GPS, WiFi), and not necessarily through sub-6GHz.} such that {the \ac{UE} evaluates the \ac{RSSI} and makes} a coarse estimation $\hat{d}$ of its position$^{{7}}$,\footnote{$^{7}$Due to our 1D deployment scenario, we indifferently refer to the distance as position in the following.} which is characterized by an estimation-error variance $\sigma_d^2$. {The \ac{UE} then relays this information to the \ac{BS} in the sub-6GHz band}.
\item When a new user arrives in the network, it associates with the \ac{BS} that provides the highest downlink received power in the sub-6GHz band. {The UE then makes a coarse initial estimation $\hat{d}$ of its position which is characterized by an estimation-error variance $\sigma_d^2$. Without loss of generality, this initial localization can be obtained by means of technologies such as sub-6 GHz band (e.g., based on RSSI or time of flight (ToF) measurements), external means such as GPS or WiFi or even with standalone \ac{mm-wave} band based distance estimation% with initial consideration of beam of large initial beamwidth
. The \ac{UE} then relays this information to the \ac{BS}.} 
\item Next, the \ac{BS} and the \ac{UE} switch to the mm-wave band. 
The \ac{UE} selects a mm-wave beam of beamwidth $\theta_{{ U}}$, initially quasi-omnidirectional (with beamwidth $\pi/2$).

\item In $\mathcal{DB}$, for each beam dictionary $k$, there exists a beam $j$ such that $d_{{{ L}}_{jk}} \leq \hat{d}\leq  d_{{{ R}}_{jk}}$. Out of all such possible beam and beamwidth pairs $j$ and $k$, the \ac{BS} selects the pair with the largest beam dictionary size (i.e., the thinnest beam) that results in a beam-selection error probability $\mathcal{P}_{{ BS},{j,k}}\left({d},\sigma_d^2\right)$ less than a threshold $\delta_{{ BS}}$. 
Mathematically, $\theta_{{ B}} = \theta_{k}$ such that
\begin{align}\label{eq:thetaB}
  k = \max(i) : \mathcal{P}_{{ BS},{j,i}}\left({d},\sigma_d^2\right) \leq \delta_{{ BS}}, \; i = 1,2, \ldots, N, d_{{{ L}}_{jk}} \leq \hat{d}\leq  d_{{{ R}}_{jk}}.
\end{align}
The expression for beam-selection error is derived in Lemma~\ref{theo:avg_BSerror}.

\item {After this step, the \ac{BS} sends downlink pilots in mm-wave band, the UE updates $\hat{d}$ and $\sigma^2_{d}$ and transmits this information in the uplink. The \ac{BS} then updates $\theta_B$ accordingly.}
%
% Mathematically,
% \begin{align}
%     \theta_{{ B}} = \theta_{k}: \mathcal{P}_{{ BS},{j,k}}\left(\hat{d}\right) \leq \delta_1; \quad k = 1,2, \ldots, N. \nonumber
% \end{align}
%
% In the case that the beam selection error condition is not satisfied, the \ac{BS} estimates a new $\hat{d}$ with variance $\sigma_d^2$ and the process is iterated with new beamwidth $\theta_{{ B}} = \theta_{k+1}$.
\item {In parallel with the ranging estimation, the UE also measures} the \ac{AoA} of the \ac{BS} signal $\hat{\psi}$, which is characterized by an estimation-error variance $\sigma_{\psi}^2$. 
First, the user sets the angle of the maximum gain equal to $\hat{\psi}$; then, it fixes $\theta_{U}$ as the thinnest beam $\theta_i$ for which the misalignment error probability $\mathcal{P}_{MA,j,k}(d, {\psi}, \sigma^2_\psi) $ is less than a threshold $\delta_{MA}$, given that the \ac{BS} selects the $j$-th beam of size $\theta_k$. Mathematically,
% \begin{align}\label{eq:thetaU}
%       \theta_U = \argmin_{\theta} \left[\mathcal{P}_{MA}(\hat{\psi}, \theta_k) \leq  \delta_{MA}\right], 
%       \text{and } 0 \leq \theta_U \leq \frac{\pi}{2}
%   \end{align}
\begin{align}\label{eq:thetaU}
    \theta_{{ U}} = \min({\theta_i}): \left[\mathcal{P}_{MA,j,k}(d, {\psi}, \sigma^2_\psi) \leq  \delta_{{ MA}}\right], 
    \text{and } 0 \leq \theta_i \leq \frac{\pi}{2}.
\end{align}
The expression for misalignment error is derived in Lemma~\ref{theo:avg_MS}.
%$\sigma^2_{\psi}$ in equation \eqref{eq:var_AoA}. \textcolor{magenta}{MC: I don't understand why it affects sigma-psi ? Do you want to say that once theta-U has been chosen, psi can be better estimated? Not ideal to refer to an equation that comes afterwards.}

\item Let $\delta_\psi$ and $\delta_d$ be the localization accuracy requirements for reliable initial access; the refinement procedure terminates when either i) the \ac{BS} beam and the \ac{UE} beam simultaneously satisfy $\sigma_d^2 \leq \delta_d$ and $\sigma_\psi^2 \leq \delta_\psi$ or ii) a maximum number of iterations is reached.

\item {When the termination conditions are not satisfied, the UE continues to measure the downlink pilots, and accordingly, the estimates of $\hat{d}$, $\sigma_d^2$,  $\hat{\psi}$, and $\sigma_\psi^2$ are updated. Following these new estimates, steps 3 and 5 are repeated for an improved initial beam selection.}

% In the subsequent iterations, the \ac{BS} uses this new beam to update the estimate $\hat{d}$ and $\sigma_d^2$ and continues this procedure until $\sigma_d^2$ is such that the beam selection error is less than $\delta_{BS}$.
% \item Following each step of \ac{BS} side beam selection and estimation of $\hat{d}$, the user estimates the \ac{AoA} of the \ac{BS} signal $\hat{\psi}$ with an error-variance $\sigma_\psi^2$ with its beam of width $\theta_{{ U}}$.
% %
% Consequently, the user decreases its beam-size, i.e., updates $\theta_{{ U}}$ to the thinnest possible beam so that the misalignment error is less than $\delta_{MA}$. In the subsequent iterations, the user employs the new beam to update the estimate of the \ac{AoA} $\hat{\psi}$ and that of the corresponding error variance $\sigma_\psi^2$.
% %
% This iterative procedure continues until $\sigma_d^2$ is such that the misalignment error is less than $\delta_{MA}$.
\end{enumerate}
We refer the reader to Fig.~\ref{fig:my_label} for a description of the steps involved in the iterative loop for the initial access. {It must be noted that the number of steps the initial beam-selection algorithm takes to terminate depends directly on the desired resolution of the localization. In other words, the more stringent the localization requirements of the initial access are, the more will be the number of steps of the initial beam-selection algorithm. Consequently, by tuning $\delta_d$ and $\delta_\psi$, the initial access delay can be controlled. There is thus an inherent trade-off between initial access delay and the accuracy of the \ac{UE} localization, which we shall discuss in Section~\ref{sec:NRD}}. 

{The proposed initial-access scheme improves the latency for establishing mm-wave connection in the system as compared to an exhaustive search solution (as we will see in the numerical results). However, in case the direct path between the user and the BS gets obstructed due to dynamic blockage, the localization performance would suffer and the system could experience beam-selection errors. In the worst case, the user might have to re-initiate the initial-access procedure. With our algorithm, this situation can be prevented by adapting the beam size using the previously stored location estimate and the current estimation accuracy. Thus, integrating the estimation accuracy (e.g., the variance of the estimation error) enables a fall-back solution that is not possible when using only location estimate. In case of using a simple exhaustive search, the entire set of beam combinations from the UE and the BS sides needs to be checked to re-estabish the connection.}
Once the initial-access process is concluded, the system initiates the data and localization phases, which are defined and optimized in the subsequent sections.

\section{Performance Characterization of the Localization and Data Phases}
\label{sec:LandD}
After the initial access phase, the {system starts the} service {phase}, which comprises two alternating phases: the {data} phase and the {localization} phase {(see Fig. \ref{fig:frame})}. In the localization phase, \ac{mm-wave} transmission is used to update the estimates of distance {$\hat{d}$} and \ac{AoA} of the signal received at the user {$\hat{\psi}$}, and potentially improve {the} localization accuracy. 
%This is essential in the case of mobile users, which require agile switching of the beams while the user moves in the coverage area of the serving BS. 
In the data phase, the \ac{UE} is served by the \ac{BS} with a \ac{mm-wave} beam, which is selected from the dictionary according to the {estimated} user location. {We propose a framework where} the radio frames are divided into flexible sub-frames {in order} to address jointly {the requirements of} localization and data services. 
%This scheme inherently presents an interesting trade-off between localization accuracy and downlink data-rates. At a first glance, it may seem that higher downlink rates can be facilitated by allotting a larger sub-frame for the data-service phase. However, assigning a small sub-frame for localization may result in a {large} error {probability} in the selection of the mm-wave beam, specially in case of thin beams, which may eventually decrease the data-rate. %Thus, it is essential to
In this section, we mathematically characterize the performance metrics of localization (\textit{v.i.z.}, position and orientation accuracy) and communication (in terms of downlink rate coverage), as a function of the resource partitioning factor $(\beta)$ and the sizes of the beams $(\theta_{{ U}},\theta_{{ B}})$, in order to optimize the radio frame design. 
%\textcolor{magenta}{MC: I feel that this introduction is too long and basically repeat what has been said in the introduction. I would write here only one or two sentences to introduce the section.}
\vspace{-0.5cm}
\subsection{Localization Phase}
\label{Localization Phase}
We model the accuracy of the {localization phase} in terms of the \ac{CRLB}s of {the} {estimated} distance of the \ac{UE} from the \ac{BS} $\hat{d}$ and of {the} \ac{AoA} $\hat{\psi}$.
\ac{CRLB} provides us with a lower bound on the variance of unbiased estimators for those two variables. 
It is {defined as} the inverse of \ac{FIM}, which measures the amount of information on {each of} the estimation variable{s} present in the observed signal, given a priori statistics for the latter.  
{Then, using these tools, we} characterize the beam-selection error $\mathcal{P}_{{ BS}}$ resulting from a distance estimation error and {we} model the misalignment $\mathcal{P}_{{ MA}}$ between the user and the \ac{BS} due to inaccuracy in the estimation of the \ac{AoA}. %Note that ultimately, aiming at characterizing mostly the average system behavior in a given urban environment, these errors are both marginalized over the a-priori spatial distribution of the BSs, and thus, of the users too. \textcolor{magenta}{MC: this last sentence is not clear.}

\subsubsection{\ac{CRLB} of the Estimation Parameters}
Let the estimates be represented by the vector $\bm \eta = [d, \psi, f_{{ R}}, f_{{ I}}]$, where $f_{{ R}}$ and $f_{{ I}}$ respectively describe the real and imaginary parts of the channel between the \ac{UE} and the serving \ac{BS}. 

\begin{lemma}
 The CRLBs for the estimation of the distance and the AoA can be written as follows:
\begin{subequations}
\begin{align}
	&\sigma^2_{d} =  \left( {\zeta}   G_U(\theta_U) G_B(\theta_B) \frac{B^2 \pi^2}{3c^2}\right)^{-1},
    \label{eq:var_d}\\
	&\sigma^2_{\psi} =  \left(\zeta  G_B(\theta_B) \left(|\dot{\bm a}_{U}^{H}(\psi) \bm w_U(\theta_U)|^2  - \frac{|\bm a_U^H(\psi) \bm w_U(\theta_U) \bm w_U^H(\theta_U) \dot{\bm a}_{U}(\psi)|^2}{G_U(\theta_U)} \right) \right)^{-1}, \label{eq:var_AoA}
\end{align}
\end{subequations}
where $\zeta = \frac{2 \text{SNR}_L B (1-\beta)  T_{F}}{G_{{ B}}(\theta_B) G_{{ U}}(\theta_U)} $, $c$ is the speed of light, $B$ is the bandwidth and
% \begin{subequations}
% \begin{align}
% 	& a_{{ t}} = G_B(\theta_B) = \bm a_{{ B}}^{{ H}}(d) \bm W_B \bm a_{{ B}}(d), 
% 	\qquad a_{r} = G_U(\theta_U) = \bm a_{U}^{H}(\psi) \bm  W_U \bm a_{U}(\psi),
% \nonumber\\
% & a_{r,1} = \dot{\bm a}_{U}^{H}(\psi) \bm W_U \dot{\bm a}_{ U}(\psi),
% 	 \quad a_{r,2} = \mathbb{R} \left\{h^* {\bm a}_{{ U}}^{{ H}}(\psi) \bm W_U \dot{\bm a}_{{ U}}(\psi) \right\},
% 	\quad a_{r,3} = \mathbb{I} \left\{ h^*{\bm a}_{{ U}}^{{ H}}(\psi) \bm  W_U \dot{\bm a}_{{ U}}(\psi) \right\},\nonumber
% \end{align}
% \end{subequations}
$\dot{\bm a}_{{ U}}(\psi) = {\partial {\bm a}_{{ U}}(\psi)}/{\partial \psi}$.  Also, $\mathbb{R} \left\{.\right\}$ and $\mathbb{I} \left\{.\right\}$ represent the real and imaginary operators.
\label{lem:CRLB}
\vspace{-0.5cm}
\end{lemma}
\begin{proof}
%See Appendix~\ref{app:CRLB}.
See \cite{Shahmansoori17, Garcia17}
\end{proof}

\begin{rmk}
\label{remark1}
The \acp{CRLB} of the estimation of the distance and the \ac{AoA} are inversely proportional to $\zeta$. Thus, the variance of the error in estimation decreases with increasing SNR$_L$ and decreasing $\beta$. Accordingly, the higher the transmit power and/or the \ac{BS} deployment density, the better the estimation performance. Similarly, larger bandwidth improves the distance estimation as it provides finer resolution for accurately analyzing the time of arrival of the received signal.
\end{rmk}

\subsubsection{Beam-Selection Error}
\label{sec:BSE}

% \sout{For a given $\theta_k$, the coverage area of the $j$-{th} beam is given by $\mathcal{C}_{j,k}  = d_{R_{jk}} - d_{L_{jk}}$, where $d_{R_{jk}}$ and $d_{L_{jk}}$ are defined in \eqref{eq:dR} and \eqref{eq:dL}, respectively.}
Without loss of generality, assume that the real position of the \ac{UE} is {$d_{L_{jk}} \leq d \leq d_{R_{jk}}$}, and accordingly, {for a given beam dictionary $k$}, the $j$-{th} beam{, whose coverage area is given by $\mathcal{C}_{j,k}  = d_{{{ R}}_{jk}} - d_{{{ L}}_{jk}}$,} should be assigned to it. 
{However, due} to ranging errors, the estimated position of the user $\hat{d}$ is distributed as $\mathcal{N}\left({d},\sigma_d^2\right)$, where $\sigma_d^2$ is defined in \eqref{eq:var_d}.
Hence, a beam-selection error occurs for the user when $\hat{{d}}$ is not inside the {correct} interval defined {by} $d_{{{ L}}_{jk}}$ and $d_{{{ R}}_{jk}}$ (see Fig. \ref{fig:BS}). 
Averaging out on the possible beams that can be selected depending on the relative positions of the typical user to BS, we have the following result.
\begin{lemma}
The probability of beam-selection error when the \ac{BS} estimates the \ac{UE} to be in the position $\hat{d}$ and selects a beam of width $\theta_k$, is computed as: %\textcolor{magenta}{MC: there is a confusion on what is the random variable and what is given. What is given here is $d$ in my opinion and thus should appear in the LHS. And the probability is obtained {\it given $d$}. Also the BS selects $j$ based on $\hat{d}$, so it cannot select $j$ if $\hat{d}\notin [dLjk; dRjk]$. In fact there is a bijective mapping bw j and d, and there is BS error when $\hat{d}$ is outside the j-th interval. }
\begin{align}\label{eq:PBS_k}
\mathcal{P}_{{ BS},{j,k}}\left(d,\sigma^2_d\right) &= \mathbb{P}\left({\hat{d}} < d_{L_{jk}}\right) + \mathbb{P}\left({\hat{d}} > d_{R_{jk}}\right) =
 1 - \mathcal{Q}\left(\frac{d_{{{ L}}_{jk}} - {d}}{{\sigma_{{ d}}}}\right) + \mathcal{Q}\left(\frac{d_{{{ R}}_{jk}} - {d}}{{\sigma_{{ d}}}}\right),
\end{align}
{where $\mathcal{Q}\left(\cdot\right)$ is the Q-function.}
Accordingly, the average beam-selection error over all the possible UE positions in case of a total number of beams $N(y)$ (where $y$ is the cell-size) with beamwidth $\theta_{{ k}}$ is given by: %\textcolor{magenta}{MC: I don't think this equation is correct. Or it depends what do you mean by $f_d$. From (14) to (15), you should sum on all intervals $\sum_j[\int_{\text{j-th interval}} P_{BS}(d)P(d|d\in \text{j-th interval})]P(d\in \text{jth interval})$}
\begin{align}
\bar{\mathcal{P}}_{{ BS}}  = \int_0^{\infty} \left(\sum_{j = 1}^{N(y)} \int_{d_{L_{jk}}}^{d_{R_{jk}}} \mathcal{P}_{{ BS}, j, k}(x,\sigma_d^2) f_{d}(x) dx \right)f_{d_a}(y) dy.
\label{eq:avg_BSerror}
\vspace*{-1cm}
\end{align}
\label{theo:avg_BSerror}
\end{lemma}
\vspace*{-1cm}
\begin{proof}
See Appendix~\ref{app:avg_BSerror}.
\end{proof}
% An important step in the proof of Theorem~\ref{theo:avg_BSerror} is the characterization of the beam-selection error in case the $i-th$ beam is selected. We will use this result in the modeling the effect of beam-selection error on the downlink data-rate, and hence we state it exclusively below:
% \begin{corollary}
% The probability of beam-selection error in case the user is estimated to be located at $x$, is computed as: 
% \begin{align}
% \mathcal{P}_{BS_i}(x) &= \mathbb{P}\left(\hat{x} < d_{L_i}\right) + \mathbb{P}\left(\hat{x} > d_{R_i}\right) \nonumber \\
% & = 1 - \mathcal{Q}\left(\frac{d_{L_i} - x}{\sigma_d^2(x)}\right) + \mathcal{Q}\left(\frac{d_{R_i} - x}{\sigma_d^2(x)}\right)
% \end{align}
% \label{cor:PBS}
% \end{corollary}

\begin{corollary}
In case of deterministic deployments, where the \acp{BS} are equispaced, \eqref{eq:avg_BSerror} becomes: %\textcolor{magenta}{MC: same remark.}
\begin{align}
&\bar{\mathcal{P}}_{{ BS}}  =  \sum_{j = 1}^{N} \int_{d_{L_{jN}}}^{d_{R_{jN}}} \mathcal{P}_{{ BS}, j, N}(x,\sigma_d^2) f_{d}(x) dx,  \quad
\mbox{ where,} \quad &N  = \ceil[\Bigg]{\frac{1}{\theta_N}\arctan\left(\frac{\frac{1}{\lambda} - d_{{ L}_{1N}}}{h_{{ B}}}\right)}.\nonumber 
\end{align}
\end{corollary}
\begin{rmk}
The wider the antenna beam, the larger is the value of $\mathcal{C}_{j,k}$. Thus, for a given distance estimation accuracy (i.e., $\sigma_d$), the beam selection error is smaller for a larger beamwidth since $\mathcal{P}_{{ BS}, j, k}(x,\sigma_d^2)$ decreases with $\mathcal{C}_{j,k}$ in \eqref{eq:PBS_k}.
On the other hand, with increasing $\theta_k$, the value of  $\sigma_{{ d}}$ increases because of the lower antenna gain. This increases the beam selection error. Overall, this results in the peaky behaviour of the beam selection error that we observe in the results.
\end{rmk}
%\textcolor{magenta}{MC: Insights can be given. For example, as the beam is thinner, $\theta_k$ is smaller, the interval length is smaller and the probability is higher. As the accuracy of the distance estimation is smaller, the proba is higher.}
\subsubsection{Misalignment Error}\label{sec:MAE}
We assume that the \ac{UE} estimates the \ac{AoA} and then sets the axis of the main lobe of its antenna to $\hat{\psi}$. %\textcolor{magenta}{MC: isnt'it $\psi$ ?}
However, in case of erroneous estimate, there exists a possibility of error in alignment of the beams (see Fig. \ref{fig:MS}). %\textcolor{magenta}{MC: Fig 4 and 5 are good but should be put sooner I think, maybe in the system model description.}
Let us assume that the user located at a distance $d$ from the BS has an \ac{AoA} $\psi$ with respect to the BS{, and that is served by the $j$-th beam of size $\theta_k$ (i.e., $\theta_B=\theta_k$)}.
Due to the noise affecting the received signal, the estimated \ac{AoA} $\hat{\psi}$ is affected by random errors.
Consequently, we assume that ${\hat{\psi}}$ is distributed as $\mathcal{N}\left({\psi},\sigma_\psi^2\right)$, where $\sigma_\psi^2$ is defined in \eqref{eq:var_AoA}.
For our analysis, we define the \ac{BS}-\ac{UE} beam pair to be misaligned, if $|\psi - \hat{\psi}|$ is larger than a threshold $\nu(\theta_B, \theta_U)$. In other words, in case the axes of the main lobe of the beams of the \ac{UE} and the BS have an angular separation larger than {the a-priori angular threshold} $\nu(\theta_B, \theta_U)$,
%{[Yet undefined threshold?]},  \textcolor{red}{are you saying that $\nu=\theta_{\textcolor{red}{{{ B}}}}$?}
we assume that the beams are misaligned. %\textcolor{magenta}{MC: it seems that in the last two sentences, the definition is not exactly the same. In the first sentence, nu doesn't appear. An important question is: does nu depend on $\theta_k$? If no: PMA doesn't play any role in the optimization problem but it is not realistic. If yes: how?}
%
% {Accordingly, we have the following result.}
\begin{lemma}
The misalignment error probability for a \ac{UE} at a distance x from the BS is given by 
%\textcolor{magenta}{MC: to be consistent with the lemma with the distance above, you can use a notation $f_{\hat{\psi}}(x)$ iso the gaussian pdf.}
%\textcolor{red}{how does $\theta_{\textcolor{red}{{{ B}}}}$ impact the following result?} \textcolor{green}{GG: $\theta_B$ changes the radiated power and hence the $\sigma_\psi^2$}:
\begin{align}
\mathcal{P}_{MA,j,k}(d, {\psi}, \sigma^2_\psi)  =   \mathbb{P}\left(\psi - \hat{\psi} \leq -\nu\right) + \mathbb{P}\left(\psi - \hat{\psi} \geq \nu\right) =  2 \mathcal{Q}\left(\frac{\nu}{{\sigma_{{ \psi}}}}\right).
%1 - \int_{-\nu/2}^{\nu/2} \frac{1}{\sqrt{2 \pi \sigma^2_{{\psi}}}} \exp\left(\frac{-\left( {\hat{\psi}-\psi} \right)^2}{2\sigma^2_{\psi}}\right) d\hat{\psi}.
\label{eq:avg_MS}
\end{align}
Then, the average misalignment probability is calculated by taking the expectation with respect to $d$ and $\psi$, i.e., $\bar{\mathcal{P}}_{MA} = \mathbb{E}_{d,{\psi}}\left[\mathcal{P}_{MA}(d,\psi,\sigma_\psi^2)\right]$,  where the distribution of $d$ is $f_d(y)$ (see Section~\ref{sec:NG}), and the distribution of ${\psi}$ is uniform between 0 and 2$\pi$.
\begin{proof}
The proof follows similar to Lemma~\ref{theo:avg_BSerror}.
\end{proof}
%\textcolor{magenta}{MC: some basic insights can be provided.}
\begin{rmk}
From \eqref{eq:avg_MS}, it can be observed that the larger the threshold for misalignment, the lower is the misalignment probability. As the threshold is directly related to the transmit and receive beamwidths, in case of wider beamwidths, the probability of misalignment is lower.
\end{rmk}
% Accordingly, the average misalignment error in the system is obtained by averaging out on the distance distribution of the user from the serving BS:
% \begin{align}
% \bar{\mathcal{P}}_{MA}  = \frac{1}{\pi}\int_0^{d_a} \int_{-\theta}^{\theta} \frac{1}{\sqrt{2 \pi \sigma^2_{\xi}(d,\xi)}} \left(\exp\left(\frac{-\xi}{2\sigma^2_{\xi}(d,\xi)}\right)\right)  f_{d}(x) d\xi dx
% \label{eq:avg_MS}
% \end{align}
\label{theo:avg_MS}
\end{lemma}

\vspace*{-1cm}
\subsection{Data Phase}
In this section, first we characterize the performance of the typical \ac{UE} considering beam-selection and misalignment errors. Then, we propose a methodology to jointly configure the split between the localization and data phases as well as the \ac{BS} beam in order to optimize data and localization performance simultaneously.
Accordingly, in the following, we first model the effective SINR coverage probability and then we define the effective user data-rate.

\subsubsection{Effective SINR Coverage Probability}
{Since the locations of the BSs are modeled as points of a 1D PPP, the locations of the users are assumed to be uniformly in the coverage area of the BSs, and the orientation of the users is assumed to be uniformly distributed between 0 and 2$\pi$, the SINR of a user is a random variable.}
The \ac{SINR} coverage probability is defined as the probability that the typical \ac{UE} receives an \ac{SINR} over a given threshold $T$. 
From the network perspective, it represents the fraction of total users  under coverage. The \ac{SINR} coverage probability is defined as the probability that the typical \ac{UE} receives an \ac{SINR} over a given threshold $T$. 
From the network perspective, it represents the fraction of total users under coverage. 
Mathematically, it is characterized in the following theorem.
\begin{theorem}
The SINR coverage probability of the typical user $\mathcal{P}_{C}\left(T,j,\theta_{{ k}},\theta_U\right)$  {served by} the $j$-th beam of width $\theta_{{ k}}$ is given by: %\textcolor{red}{@Gourab: In the below eqn, should it be $\sigma^2_d(x)?$}\textcolor{violet}{check $\sigma_x^2$ and harmonize; $\sigma^2_\psi$ depends on $x$ and $\psi$ but we have never used this notation before}:
\begin{align}
\label{eq:cov}
&\mathcal{P}_{C}\left(T,j,\theta_{{ k}},\theta_U\right) = \int_{0}^{2 \pi} \int_{d_{L_{j,k}}}^{d_{R_{j,k}}} \left[\mathcal{P}_{{ BS},{j,k}}(x,\sigma_d^2)\mathcal{T}_{BS}(x,T) + \left(1 - \mathcal{P}_{{ BS},j,k}(x,\sigma^2_d)\right)  \nonumber \right.\\
& \left.\left(\left(1 - \mathcal{P}_{{MA},j,k}(x, \psi,\sigma^2_\psi )\right)\mathcal{T}_0(x,T) + \mathcal{P}_{{MA},j,k}(x,\psi,\sigma^2_\psi) \mathcal{T}_{MA}(x,T)\right)\right]  f_d(x) f(\psi) dx d\psi
%\partial x \partial \psi,
\end{align}
\vspace{-0.3cm}
where
\vspace{-0.3cm}
\begin{align}
&\mathcal{T}_0(x,T) =  \sum\limits_{n = 1}^{N_L}(-1)^{n+1} \; \binom {N_L}n \exp\left(-\left( \frac{n \eta_L T {N_0}}{P_t K \gamma_{{ B}}(\theta_{{ k}})\gamma_{{ U}}(\theta_{{ U}}) z^{-\alpha_{L}}} + \mathcal{A}_{L0}\left(x,T\right) + \mathcal{A}_{N0}\left(x,T\right) \right) \right), \nonumber\\
&\mathcal{T}_{MA}(x,T) =  \sum\limits_{n = 1}^{N_L}(-1)^{n+1} \; \binom {N_L}n \exp\left(-\left( \frac{n \eta_L T {N_0}}{P_t K \gamma_{{ B}}(\theta_{{ k}}) g z^{-\alpha_{L}}} + \mathcal{A}_{LMA}\left(x,T\right) + \mathcal{A}_{NMA}\left(x,T\right) \right) \right), \nonumber\\
&\mathcal{T}_{BS}(x,T) = \sum\limits_{n = 1}^{N_L}(-1)^{n+1} \; \binom {N_L}n \exp\left(-\left( \frac{n \eta_L T {N_0}}{P_t K g^2 z^{-\alpha_{L}}} + \mathcal{A}_{LBS}\left(x,T\right) + \mathcal{A}_{NBS}\left(x,T\right) \right) \right)\nonumber,
\end{align}
in which $z = \sqrt{x^2 + h_B^2}$, $\sigma_d^2$ is a function of $x$, $\sigma_\psi^2$ is a function of $x$ and $\psi$, and
%\textcolor{magenta}{MC: you should be able to define 2 generic A functions depending on the factors used in the fraction and the integral bounds. Maybe it can help for the readability.}%\textcolor{red}{the following equation have to be reshaped to fit in the page width.}
\begin{align}
&\scalemath{.75}{\mathcal{A}_{L0}(x,T) = 2\lambda \int_{x}^{d_S} 1 - \frac{1}{\left(1 + \frac{\eta_L Tg^2q^{-\alpha_L} 2\lambda y}{N_L \gamma_{B}(\theta_{{ k}})\gamma_{{ U}}(\theta_{{ U}})z^{-\alpha_L}}\right)^{N_L}}  dy,} \hphantom{.} \scalemath{.75}{\mathcal{A}_{N0}(x,T) = 2\lambda \int_{d_S}^{\infty} 1 - \frac{1}{1 + \left(\frac{\eta_L Tg^2q^{-\alpha_N} 2\lambda (y-d_S)}{N_N \gamma_{{ B}}(\theta_{{ k}})\gamma_{{ U}}(\theta_{{ U}})z^{-\alpha_L} + n \eta_L Tg^2q^{-\alpha_N}}\right)^{N_N}} dy} \nonumber, \\
&\scalemath{.75}{\mathcal{A}_{LMA}(x,T) = 2\lambda \int_{x}^{d_S} 1 - \frac{1}{\left(1 + \frac{\eta_L Tgq^{-\alpha_L} 2\lambda y}{N_L \gamma_{B}(\theta_{{ k}})z^{-\alpha_L}}\right)^{N_L}}  dy,} \hphantom{*a}\scalemath{.75}{\mathcal{A}_{NMA}(x,T) = 2\lambda \int_{d_S}^{\infty} 1 - \frac{1}{1 + \left(\frac{\eta_L Tgq^{-\alpha_N} 2\lambda (y-d_S)}{N_N \gamma_{{ B}}(\theta_{{ k}})z^{-\alpha_L} + n \eta_L Tg^2q^{-\alpha_N}}\right)^{N_N}} dy} \nonumber \\
&\scalemath{.75}{\mathcal{A}_{LBS}(x,T) = 2\lambda \int_{x}^{d_S} 1 - \frac{1}{\left(1 + \frac{\eta_L Tq^{-\alpha_L} 2\lambda y}{N_L z^{-\alpha_L}}\right)^{N_L}}  dy,}
\hphantom{****.a}\scalemath{.75}{\mathcal{A}_{NBS}(x,T) = 2\lambda \int_{d_S}^{\infty} 1 - \frac{1}{1 + \left(\frac{\eta_L Tq^{-\alpha_N} 2\lambda (y-d_S)}{N_N z^{-\alpha_L} + n \eta_L Tg^2q^{-\alpha_N}}\right)^{N_N}} dy} \nonumber,
\end{align}
and $q = \sqrt{y^2 + h_B^2}$.
\label{theo:CovP_LOS}
\end{theorem}
\begin{proof}
See Appendix~\ref{app:CovP_LOS}. 
\end{proof}

% \begin{subequations}
% \begin{align}
% \scalemath{.87}{\mathcal{A}_{L0}(x,T) = \int_{x}^{d_S}\frac{Tg^2y^{-\alpha_L} 2\lambda y}{ \gamma_{{ B}}(\theta_{{ k}})\gamma_{{ U}}(\theta_{{ U}})x^{-\alpha_L} + Tg^2y^{-\alpha_L}}  dy; \quad  \mathcal{A}_{N0}(x,T) = \int_{d_S}^{\infty}\frac{Tg^2y^{-\alpha_N} 2\lambda (y-d_S)}{ \gamma_{{ B}}(\theta_{{ k}})\gamma_{{ U}}(\theta_{{ U}})x^{-\alpha_L} + Tg^2y^{-\alpha_N}} dy} \nonumber
% \end{align}
% \begin{align}
% \scalemath{.87}{\mathcal{A}_{LMA}(x,T) = \int_{x}^{d_S}\frac{Tgy^{-\alpha_L}  2\lambda y}{ \gamma_{{ B}}(\theta_{{ k}})x^{-\alpha_L} + Tgy^{-\alpha_L}} dy; \quad  \mathcal{A}_{NMA}(x,T) = \int_{d_S}^{\infty}\frac{Tgy^{-\alpha_N} 2\lambda(y-d_S) }{ \gamma_{{ B}}(\theta_{{ k}})x^{-\alpha_L} + Tgy^{-\alpha_N}} dy }\nonumber
% \end{align}
% \begin{align}
% \scalemath{.87}{\mathcal{A}_{LBS}(x,T) = \int_{x}^{d_S}\frac{Ty^{-\alpha_L} 2\lambda y}{x^{-\alpha_L} + Ty^{-\alpha_L}}  dy; \quad  \mathcal{A}_{NBS}(x,T) = \int_{d_S}^{\infty}\frac{Ty^{-\alpha_N} 2\lambda(y-d_S)}{x^{-\alpha_L} + Ty^{-\alpha_N}}  dy }\nonumber.
% \end{align}
% \end{subequations}

%\textcolor{magenta}{MC: I think that the T functions are increasing when $\theta_k$ is smaller, while for the terms 1-P it is the reverse effect. This should be mentioned somewhere either here or when the optimization problem is formulated. }
In \eqref{eq:cov}, the term $\mathcal{T}_0(x,T)$ corresponds to the case in which there is no beam-selection error as well as no misalignment. 
In this case, we have $G_{B}(\theta_{{ k}}) =  \gamma_{{ B}}(\theta_{{ k}})$ and $G_{U}(\theta_{{ U}}) =  \gamma_{{ U}}(\theta_{{ U}})$ resulting in a high coverage probability. The term $\mathcal{T}_{MA}(x,T)$ represents the case where there is no beam-selection error, but the \ac{BS}-user beam pair suffers from misalignment. Here the coverage probability decreases as compared to $\mathcal{T}_0(x,T)$ although $G_{B}(\theta_{{ B}})$ remains the same, since here we have $G_{U}(\theta_{{ U}}) = g$. Finally, the term $\mathcal{T}_{BS}(x,T)$ refers to the case when there is a beam-selection error. It must be noted that according to our assumption, in the case of beam-selection error, we assume that the beams are always misaligned. Here we have $G_{B}(\theta_{{ k}}) = G_{U}(\theta_{{ U}}) = g$.
In case of exhaustive-search, the users will not suffer from beam-selection or misalignment errors, i.e., for exhaustive-search, in \eqref{eq:cov} we have $\mathcal{P}_{{ BS},{j,k}}(x,\sigma_x^2) = 0$ and $\mathcal{P}_{{MA},j,k}(x,\psi,\sigma^2_\psi) = 0$. Accordingly, the users will experience a better SINR, as discussed in the following proposition:
\begin{proposition}
For a given value of $\theta_k$ and $\theta_U$, an exhaustive-search based initial access algorithm will suffer from no beam-selection error and no misalignment error. Consequently, the SINR coverage probability for an exhaustive search algorithm is given by:
\begin{align}
    \mathcal{P}_{C}\left(T,j,\theta_{{ k}},\theta_U\right) = \int_{0}^{2 \pi} \int_{d_{L_{j,k}}}^{d_{R_{j,k}}} \mathcal{T}_0(x,T)  f_d(x) f(\psi) dx d\psi
    \label{eq:exhaust}
\end{align}
\end{proposition}
\begin{corollary}
The overall SINR coverage probability, considering all the $N$ beams of size $\theta_{{ k}}$ is:
\begin{align}
\mathcal{\bar{P}}_{C}(T,\theta_{{ k}},\theta_U) = \mathbb{E}_{d_a}\left[\sum_{j=1}^{N(d_a)}\mathcal{P}_{C}\left(T,j,\theta_{{ k}},\theta_U\right)\right],
\end{align}
\end{corollary}
where the expectation is taken with respect to the inter-BS distance $d_a$ given by \eqref{eq:IS-distribution2}.
% \begin{corollary}
% Accordingly, in case of NLOS association without beam selection error and the 
% \end{corollary}

\subsubsection{Effective Rate Coverage Probability}
Let $B$ denote the system bandwidth and $T_I$ the duration of the initial access procedure. 
As the data phase uses $\beta$ fraction of the total \sout{time-frequency} resources in the service phase $T_F$, we can compute the probability $\mathcal{P}_R(r_0, \beta, \theta_{{ k}},\theta_U)$ that the effective rate is above given threshold $r_0$ as below. %\textcolor{magenta}{MC: clearly when $\beta$ increases, rate increases.}
\vspace*{-0.3cm}
\begin{lemma}
For a given SINR coverage probability, the effective rate coverage probability is given by $\mathcal{\bar{P}}_C\left( 2^{\frac{r_0(T_I+T_F)}{\beta T_F B}} -1, \theta_{{ k}},\theta_U\right)$, where $r_0$ is the target rate threshold.
\end{lemma}
\begin{proof}
\vspace*{-0.3cm}
\begin{align}
\vspace*{-0.3cm}
\mathcal{\bar{P}}_R(r_0, \beta, \theta_{{ k}},\theta_U) &= \mathbb{P}\left(\frac{\beta T_F}{T_I + T_F}B\log_2\left(1 + SINR_C\right) \geq r_0\right) = \mathbb{P} \left(SINR_C \geq 2^{\frac{r_0(T_I+T_F)}{\beta T_F B}}  -1\right) \nonumber\nonumber \\ & =  \mathcal{\bar{P}}_C\left( 2^{\frac{r_0(T_I+T_F)}{\beta T_F B}}  -1, \theta_{{ k}},\theta_U\right).
\label{eq:rate}
\end{align}
\vspace*{-1cm}
\end{proof}
\subsection{Joint Optimization of the Transmit Beamwidth and Radio Frame Structure}
%\label{sec:Algo}

Given the characterization of the effective rate coverage probability, we present a schematic for selecting, at the BS, the optimal beam from the designed beam dictionary.
{The proposed schematic is presented in the form of a two-stage optimization problem as:
	\begin{equation}
	 \theta^* = \underset{\theta_k}{argmax} \begin{bmatrix} \underset{\beta}{max} \quad &\mathcal{\bar{P}}_R(r_0,\beta,\theta_{{ k}},\theta_U)\\
     \text{subject to} \qquad	&\bar{\mathcal{P}}_{BS} \leq \epsilon \\
     & \bar{\mathcal{P}}_{MA} \leq \epsilon'
     \end{bmatrix}
	\label{initial}.    
\end{equation}
In the first step, for a given $\theta_k$, we select the value of $\beta_k^*$ that maximizes the effective rate coverage probability subject to apriori constraints $\epsilon$ and $\epsilon'$, on the beam-selection and the misalignment errors, respectively. These constraints are system parameters which are governed by the accuracy requirements for the localization service.
%\textcolor{red}{In the se}
%
In the subsequent frames, based on the new measurements, the estimates ($\hat{d}$ and $\hat{\psi}$) are updated and the measurement error variances ($\sigma^2_d$ and $\sigma^2_\psi$) change. Accordingly, the BS beamwidth $\theta^*$ can be further updated by using \eqref{initial}.}

{We emphasize that the optimal beamwidth thus calculated is different from the adaptive beamwidth value evaluated in the initial access phase. The former is calculated offline to maximize the data-rate given a set of system parameters, whereas, the latter is the beamwidth adapted to achieve the required resolution in terms of distance and orientation.}

\section{Numerical Results and Discussion}
\label{sec:NRD}
Now we present numerical results related to the initial beam-selection and the localization-communication trade-offs developed in this paper. {The parameter values are given in Table 1. The numerical results follow the analytical expressions derived in this paper, where the beam-selection and misalignment errors are characterized by (16) and (18), respectively. The errors are incorporated into the SINR coverage probability expressions as derived in Theorem 1. Leveraging this, the rate coverage probability follows (21). }
%In the following subsections, we study the performance of the three phases depicted in Fig. 1 and the optimal resource partitioning and beamwidth selection scheme.}

\vspace*{-0.5cm}
\subsection{Performance of the Initial Access Phase}
First, let us discuss the performance of the initial beam-selection strategy developed in Section \ref{sec:IA}. In Fig.~\ref{fig:initial2}, we plot the enhancement in positioning resolution (characterized as the variance of the ranging error) with increasing the number of steps of our initial beam-selection algorithm. %{define what positioning resolution is; there is no enhancement in Fig.~\ref{fig:initial2}.}
Here, we have assumed that $\delta_d = 0.01$m is the minimum resolution required to provide mm-wave date-service. 
As expected (see Remark 1), we note that for denser small cell deployments (e.g., $\lambda = 0.1$ m$^{-1}$) the algorithm stops at a lower number (here 4) of iterations, as compared to the sparser deployment scenarios. 
%{lower with respect to what? and why dense deployment leads to fast convergence}
%
As the deployment becomes sparser (e.g., $\lambda = 0.01$ m$^{-1}$), {a larger number of steps is} required for the initial access procedure. This is precisely due to the fact that for denser deployments, $\text{SNR}_L$ increases. Accordingly,  a  larger  beamwidth  is  sufficient  and  hence,  a  lower number  of iterations are required to meet the localization requirements.

%\textcolor{magenta}{MC: I would speak somewhere about the delay vs resolution tradeoff.}
%
% \begin{figure}
% \centering
%     \subfloat[]
%     {\includegraphics[width = 8cm, height = 5cm]{INIT.eps}
%     \label{fig:initial2}}
%     \hfill
%     \subfloat[]
%     {\includegraphics[width = 8cm, height = 5cm]{initial_final.eps}
%     \label{fig:delay}}
%     \caption{(a) Resolution in the $n$-th step of the bound-based initial access strategy for different deployment densities and (b) Comparison of the delay in initial access of our bound-based strategy to the iterative and exhaustive search strategies.}
% \vspace*{-1cm} \end{figure}
\begin{figure}
\centering
  \subfloat[]
    {\includegraphics[height = 4cm, width = 8cm]{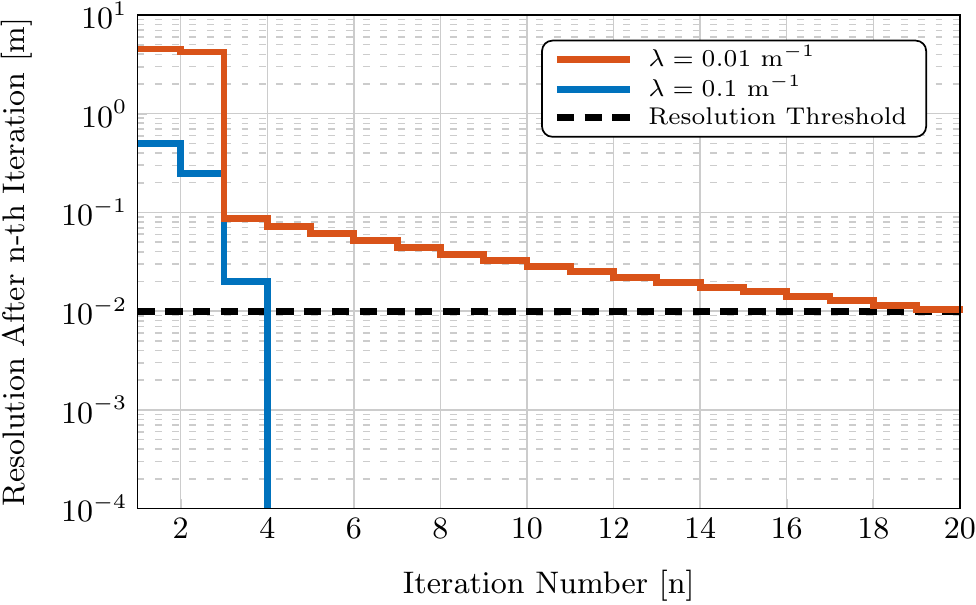}
    \label{fig:initial2}}
    \hfill
  \subfloat[]
    {\includegraphics[height = 4cm, width = 8cm]{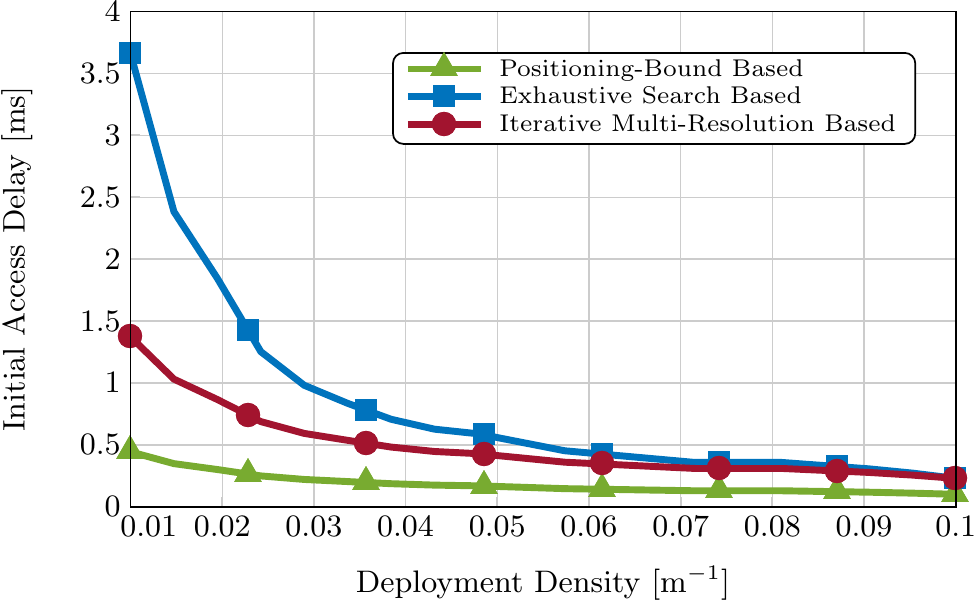}
    \label{fig:delay}}
    \caption{(a) Resolution in the $n$-th step of the localization-based initial access strategy for different deployment densities, and (b) Comparison of the delay in initial access of our localization-based strategy to the iterative and exhaustive search strategies.}
    \vspace*{-0.5cm}
\end{figure}

\begin{figure}
    \centering
  %      \subfloat[]
    {\includegraphics[height = 4cm, width = 8cm]{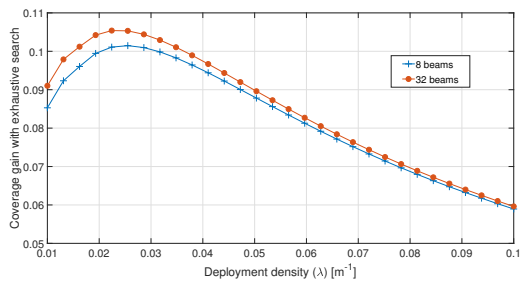}
    }
    \caption{ Gain in SINR coverage with an exhaustive search based initial-access algorithm for two beam dictionary sizes.}
    \label{fig:exhaust}
    \vspace*{-0.5cm}
\end{figure}

Then, in Fig.~\ref{fig:delay} we compare the initial access delay of the proposed localization-bound based strategy with the one achieved by two well-known beam-sweeping solutions: exhaustive search and iterative search~\cite{li2017design,michelusi2018optimal}.
For the exhaustive search, we consider the beamwidth of the BS and the UE to be fixed and equal to $\theta_B$ and $\theta_U$, respectively. Thus, the \ac{BS} and the \ac{UE} go through all the possible $\frac{2\pi}{\theta_B} \times \frac{2\pi}{\theta_U}$ beam combinations to select the beam pair that maximizes the SNR. The exhaustive search has been adopted in the standards IEEE 802.15.3c and IEEE 802.11ad~\cite{baykas2011ieee,anastasi2003ieee}. On the other hand, for the case of iterative search (similar to bisection search in~\cite{hussain2017throughput}), we assume that the \ac{BS} initiates the procedure with $k = 2$ while the user uses an omni-directional beam. 
Out of the two possible beams, the \ac{BS} identifies the beam that results in the highest downlink SNR and changes its search space to the region covered by that beam. Then, the \ac{BS} changes its beam size to a thinner one (of dictionary $k = 4$) and uses 2 out of the 4 beams from the dictionary which lie within the modified search space. We assume that the initial beam-selection phase terminates when this process chooses the same beamwidth $\theta_{{B}}$ selected by our algorithm. Thereafter, the \ac{BS} fixes the selected beam and the \ac{UE} carries out the same procedure for obtaining the user side beam. In our system, similar to~\cite{li2017design}, we assume that i) one OFDM symbol length (including cyclic prefix) is 14.3~$\mu$s, ii) each synchronization signal occupies only one OFDM symbol, and iii) the beam reference signal is also transmitted in the same symbol to uniquely identify the beam index. Clearly, our strategy provides considerably faster initial access precisely due to the smaller number of steps than those required by the exhaustive and iterative search based schemes. %\textcolor{magenta}{MC: very interesting but the conditions are not sufficiently detailed, we can discuss during a meeting. For the OFDM symbol duration incl CP, I have 17.9mus for 60kHz SCS and 8.92mus for 120kHz in 5G NR. Initial access delays are too short to be realistic.}
% \begin{figure}
% \centering
% \includegraphics[width = 8cm, height = 5cm]{initial3.eps}
% \caption{Probability of initial access error with respect to the deployment densities for different resolution requirements {the x-axis data is missing; what is the initial access error?}.}
% \label{fig:initial3}
% \vspace*{-1cm} \end{figure}

% In Fig.~\ref{fig:initial3} we compare the beam-selection error {eq?} during \textcolor{red}{the} initial access for different resolution requirements. For a higher resolution requirement, the initial access procedure requires a higher number of steps (see Fig.~\ref{fig:initial2}). Since the beam-selection error in \textcolor{red}{the} initial access cumulatively increases with the number of steps, the initial access error is higher in case of a higher resolution requirement.
% {the other two strategies do not have error, isn't?; this tradeoff should be discussed.}

The number of iterations our algorithm takes to terminate is a direct measure of the delay in the initial beam-selection procedure. Specifically, we assume that this delay is computed as the product of the sum of the required number of steps at the BS side and the UE side and the duration of one OFDM symbol. In Fig.~\ref{fig:initial2}, with $\lambda =$ 0.01 $\mbox{m}^{-1}$, we observe that for a required $\delta_d=$ 0.01 m, our algorithm terminates in 20 steps, which corresponds to a delay of about 5.7 ms. Wheras, if the positioning requirement was specified to be 0.1 m, the algorithm would have terminated in 3 steps, which corresponds to a initial beam-selection delay of about 1 ms. Thus, there exists a fundamental trade-off between the localization requirement ($\delta_d$ and $\delta_\psi$) and the delay in the initial beam-selection.

{For a fair comparison, we emphasize that classical algorithms such as the exhaustive search do not suffer from beam-selection and misalignment errors. This is shown in Fig.~\ref{fig:exhaust}, where we plot the gain in SINR coverage with an exhaustive search based initial-access algorithm as compared to our proposed algorithm for two beam dictionary sizes. We observe that with a large number of beams, the SINR gain increases. This is precisely because a large beam dictionary size leads to smaller beam coverage, which in turn increases the beam selection error. More interestingly, we see that for dense deployment of BSs, the gain drops dramatically as the beam-selection and misalignment errors with the proposed initial-access scheme are limited.}

\vspace*{-0.5cm}
\subsection{Performance of the Localization Phase}
The reduction in the initial beam-selection delay with the proposed algorithm is naturally associated with localization errors, which we discuss in this section.
% \begin{figure}
% \centering
% \subfloat[]
% {\includegraphics[width = 8cm, height = 5cm]{PBS_FINAL.eps}
% \label{fig:SBE}}
% \hfill
% \subfloat[]
% {\includegraphics[width = 8cm, height = 5cm]{PMA_FINAL.eps}
% \label{fig:MAE}}
% \caption{(a) Probability of beam-selection error vs the resource partitioning factor for different beam dictionary sizes and (b) Probability of misalignment vs the resource partitioning factor for different beam dictionary sizes.} 
% \vspace*{-1cm} \end{figure}

%% \begin{figure}
%% \centering
%% \subfloat[]
%% {\includegraphics[width = 8cm]{PBA.eps}
%% \label{fig:PBA}}
%% \hfill
%% \subfloat[]
%% {\includegraphics[width = 8cm]{PMA.eps}
%% \label{fig:PMA}}
%% \caption{(a) Probability of beam selection error vs the beam dictionary size for different antenna gains; (b) Probability of misalignment error vs the beam dictionary size for different antenna gains.}
%% \vspace*{-1cm} \end{figure}

\begin{figure}
\centering
\subfloat[]
{\includegraphics[height = 4cm, width = 8cm]{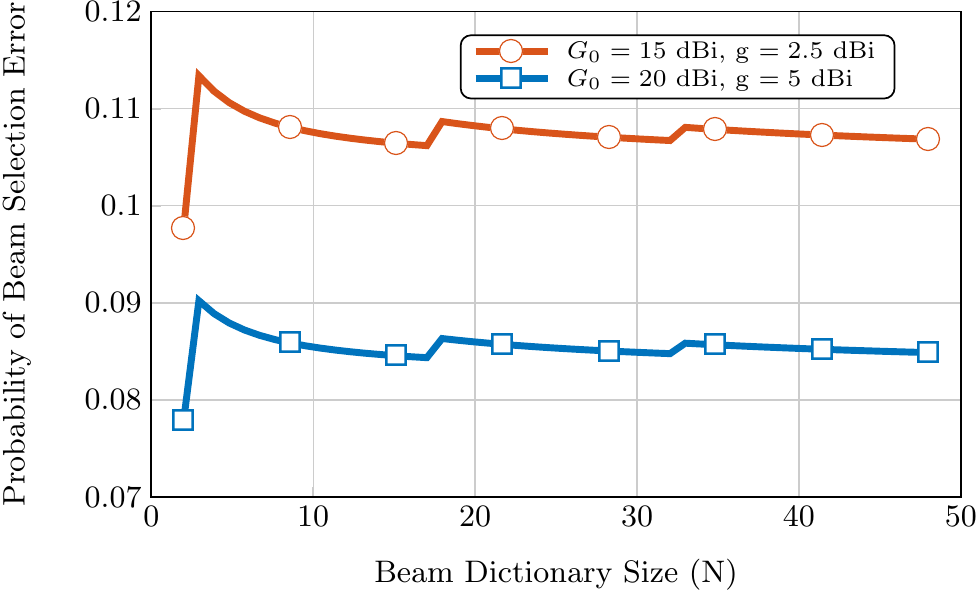}
\label{fig:PBA}}
\hfill
\subfloat[]
{\includegraphics[height = 4cm, width = 8cm]{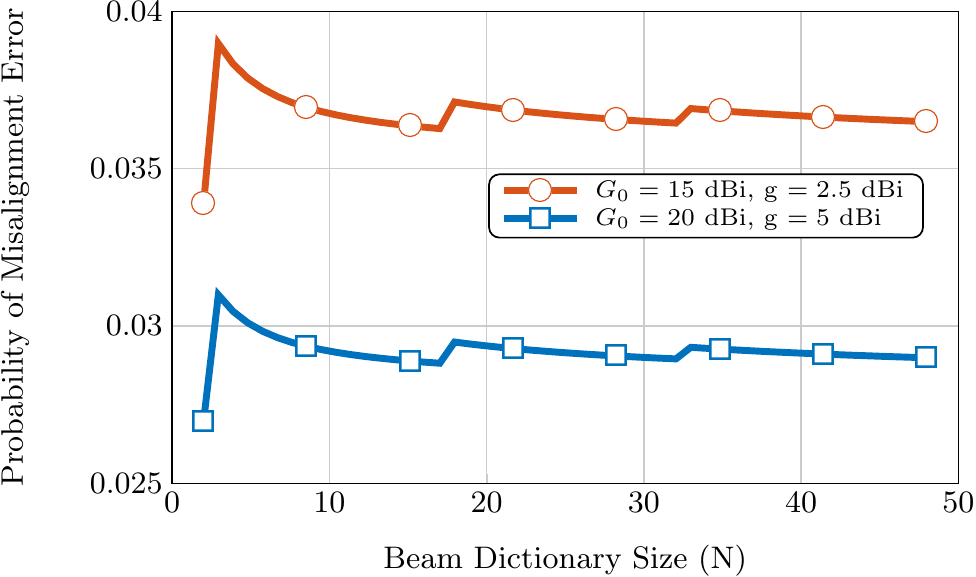}
\label{fig:PMA}}
\caption{(a) Probability of beam selection error vs the beam dictionary size for different antenna gains; (b) Probability of misalignment error vs the beam dictionary size for different antenna gains.}
\vspace*{-1cm} \end{figure}

In Fig.~\ref{fig:PBA} we plot the beam selection error as a function of the beam-dictionary size ($N$) for different antenna gains. It must be noted that the effect of a larger beamwidth on the beam-selection error is non-trivial (see Remark 2). Larger beamwidth results in a lower radiated power, which leads to a higher \ac{CRLB} for distance estimation, which may lead to a higher beam-selection error. However, a larger beamwidth also corresponds to a larger geographical area covered on ground by the beam (i.e., larger $\mathcal{C}_{k,j}$), which leads to a lower beam-selection error.
%\textcolor{magenta}{MC: You should first explain the contradictory effects: if beam are large, there is less chance to make a mistake but received power is less and so localization accuracy is also less, which leads to higher errors}

As expected, the beam-selection error is minimized for $N = 1$, when a single beam encompasses all the cell coverage area of the \ac{BS}. The beam-selection error would occur only when using a distance-based cell selection scheme, if the actual position of the \ac{UE} is outside the coverage area of the serving \ac{BS}. For $N \geq 2$, interestingly, we observe a stepped behavior of the probability of beam selection error with respect to the beam dictionary size. The beam selection error gradually decreases with increasing beam dictionary size due to the increasing antenna gain (see eq. \eqref{eq:Gain}). This behaviour continues until a certain value of beam dictionary size, where the beam width becomes so thin that the probability that the user lies outside the beam coverage area is high. This results in an increase in the probability of beam-selection error, which then gradually decreases, when increasing the beam dictionary size, and so on. %\textcolor{magenta}{MC: I think that you don't really explain the step behavior: when N increases the beams are thinner and thinner so why the error decreases after the steps?} 
This brings forth an important characteristic of the system: for achieving a given beam selection error performance, multiple beam sizes can exist. This is precisely because of the fact that with the decreasing size of the beams, two conflicting phenomena occur: i) an  improvement in the estimation performance owing to larger antenna gain and ii) a reduction of the geographical area covered by each beam.

Fig.~\ref{fig:PMA} shows that the beam misalignment probability has the same peaky trend of the beam selection error with respect to the beam dictionary size. Specifically, the misalignment probability gradually decreases with increasing $N$ until a certain value beyond which the beam becomes so thin that the misalignment error increases.

\vspace*{-0.5cm}
\subsection{Localization Data-Rate Trade-off}

\begin{figure}
    \centering
    \includegraphics[width = 0.6\textwidth]{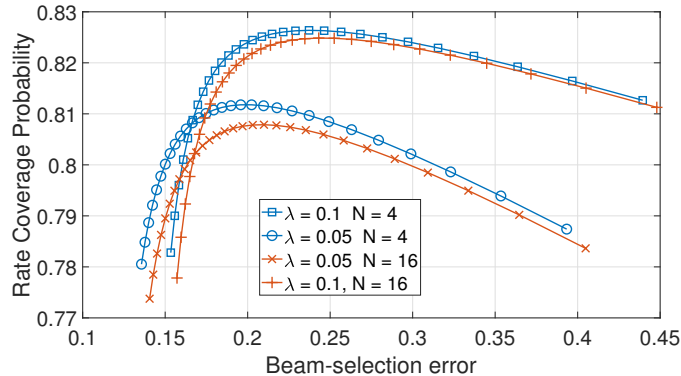}
    \caption{Rate coverage probability vs the beam-selection error.}
    \label{fig:modified_rate}
    \vspace*{-1cm}
\end{figure}

{In Fig.~\ref{fig:modified_rate} we plot the rate coverage probability of the typical user with respect to the average beam-selection error, $\bar{\mathcal{P}}_{BS}$ for different beam dictionary sizes. We tune the value of $\bar{\mathcal{P}}_{BS}$ by modifying the value of $\beta$. We observe that for all beam dictionary sizes, as the value of $\bar{\mathcal{P}}_{BS}$ increases (i.e., as the localization estimation performance degrades), the rate coverage probability is initially improved. Thereafter, it reaches an optimal value for a certain $\bar{\mathcal{P}}_{BS}$ and decreases on further increasing the value of $\bar{\mathcal{P}}_{BS}$.
This highlights the non-trivial trade-off between the localization and the data-rate performance in our system. {{This is all the more complex} as the optimal value of $\beta$ (and hence the rate coverage probability) depends on both the BS deployment density and the dictionary size.} To achieve very low values of $\bar{\mathcal{P}}_{BS}$, sufficient resources need to be allotted for the localization phase thus leading to efficient beam-selection and beam-alignment. A small increase in the value of $\bar{\mathcal{P}}_{BS}$ does not result in a large degradation of the localization performance but, in contrast, enhances the data-rate as more resources are assigned to the data-communication phase.  
However, further increasing the value of $\bar{\mathcal{P}}_{BS}$ after a certain $\beta$ (i.e., $\beta^*$) deteriorates the rate coverage. This is because poor localization leads to a high beam selection errors. As a result, the effective antenna gains at the transmitter and receiver sides decrease, which directly reduces the useful received signal power, while the interference power remains same. Overall, this leads to limited rate performance.
Another interesting observation in this figure is that in order to achieve the same coverage performance, the {beam-selection error is slightly larger in case} of larger beam-dictionaries. This is due to the thinner beams in larger beam-dictionaries, which increase the probabilities that the users lie outside the serving beam.}

\vspace*{-0.5cm}
\subsection{Rate Coverage Performance and Trends}
In Figs.~\ref{fig:Rate1} and \ref{fig:Rate3} we plot the rate coverage probability with respect to the resource partitioning factor $\beta$ varying the antenna gain parameter $G_0$ and the BS deployment density. First, we note again that there exists an optimal $\beta^*$ for each beam dictionary size, for which the rate coverage probability is maximized. More interestingly, the value of $\beta^*$ is not unique and is dependent not only on the dictionary size but also on the system parameters such as antenna gains. From Fig.~\ref{fig:Rate1} we can see that the optimum value of $\beta$ decreases for higher $N$, i.e., thinner beamwidth. This is because with thinner beamwidth, the localization resources should be increased to limit the probability that the \ac{UE} lies outside the coverage area of the beam. 

%Comparing Fig.~\ref{fig:Rate1} with Fig.~\ref{fig:Rate2}, we notice that for a sparser deployment of base stations (10 per km), $\beta^*$ is closer to 0 than the case with denser deployments (100 per km). This is precisely due to the fact that for sparser deployment of \acp{BS}, more resources should be allotted to the localization phase for accurate position information. Whereas for denser deployments, the close proximity of the serving \ac{BS} to the user itself leads to satisfactory localization performance, leading to a better data-communication performance. %This can also be observed with while comparing the trends in rate coverage when $\beta \to 0$ and $\beta \to 1$. 
%For $\beta = 0.1$, i.e., with very precise localization, the sparser deployment case fares better than the denser one. This is due to the fact that given that beam selection error and misalignment does not occur (with accurate localization), the sparser case enjoys lower interference from the neighboring \acp{BS}. The trend reverses when $\beta$ is closer to 1. That is, when the resource allotted to the localization phase is inadequate, the denser deployment, with higher localization performance based on closer proximity leads to a better rate coverage.
%We notice that for sparser deployment of base stations (10 per km), the optimal rate coverage with thinner beams ({$N$} = 16) is lower ($ < 0.82$) than the case with denser deployments (100 per km), where the coverage reaches 0.825. \textcolor{magenta}{MC: small gain whereas the number of BSs has been multiplied by 10!}

When the antenna gain is smaller ($G_0 = 7.5$ dBi), we see in Fig.~\ref{fig:Rate3} that the rate coverage (at 1 Mbps contrary to 100 Mbps as before) increases with $\beta$. With {$G_0 = 15$} dBi, the positioning accuracy is limited (for any value of $\beta$), while increasing $\beta$ simply increases the communication resources, thereby augmenting the coverage. In this case, a smaller beamwidth (with  {$N$} = 16) provides better coverage than a larger beamwidth (with {$N$} = 4), since with limited localization accuracy, the rate coverage simply increases with decreasing $\theta$ due to higher radiated power.
% \begin{figure}
% \centering
% \subfloat[]
% {\includegraphics[width = 8cm]{c1.eps}
% \label{fig:Rate1}}
% %\hfill
% % \subfloat[]
% % {\includegraphics[width = 5cm]{c2.eps}
% % \label{fig:Rate2}}
% \subfloat[]
% {\includegraphics[width = 8cm]{c3.eps}
% \label{fig:Rate3}}
% \caption{Rate coverage probability versus the resource partitioning factor for different beam dictionary sizes.}
% %{[BD: (i) Can't we have the same order while exposing the parameters in the legend. It's not so easy to track the main differences btw. (a) and (b) in terms of settings at very first sight ? (ii) Why stating $G$ on Fig. (b) ?]}}
% \vspace*{-1cm} \end{figure}
\begin{figure}
\centering
\subfloat[]
{\includegraphics[height = 4cm, width = 8cm]{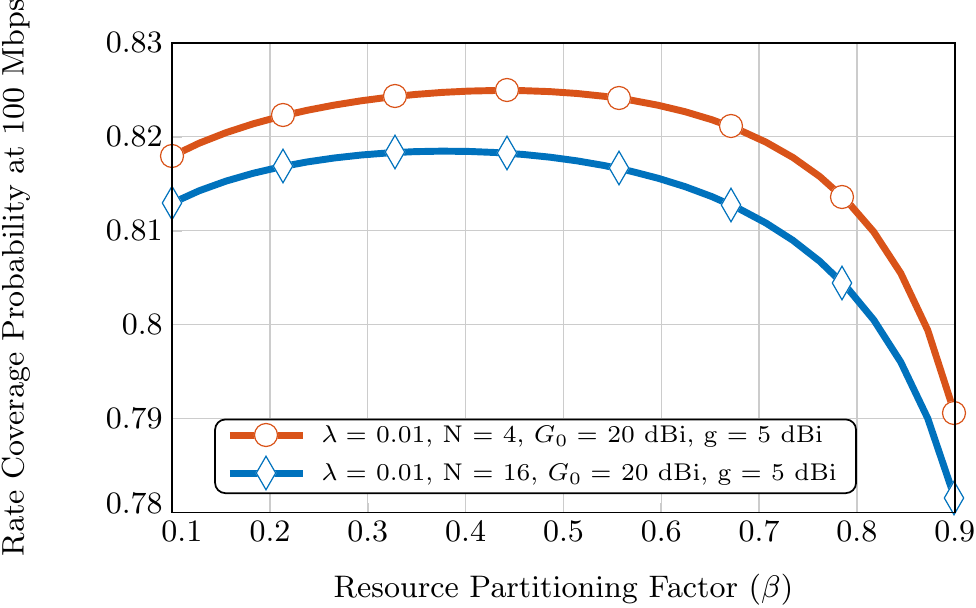}
\label{fig:Rate1}}
%\hfill
% \subfloat[]
% {\includegraphics[width = 5cm]{c2.eps}
% \label{fig:Rate2}}
\subfloat[]
{\includegraphics[height = 4cm, width = 8cm]{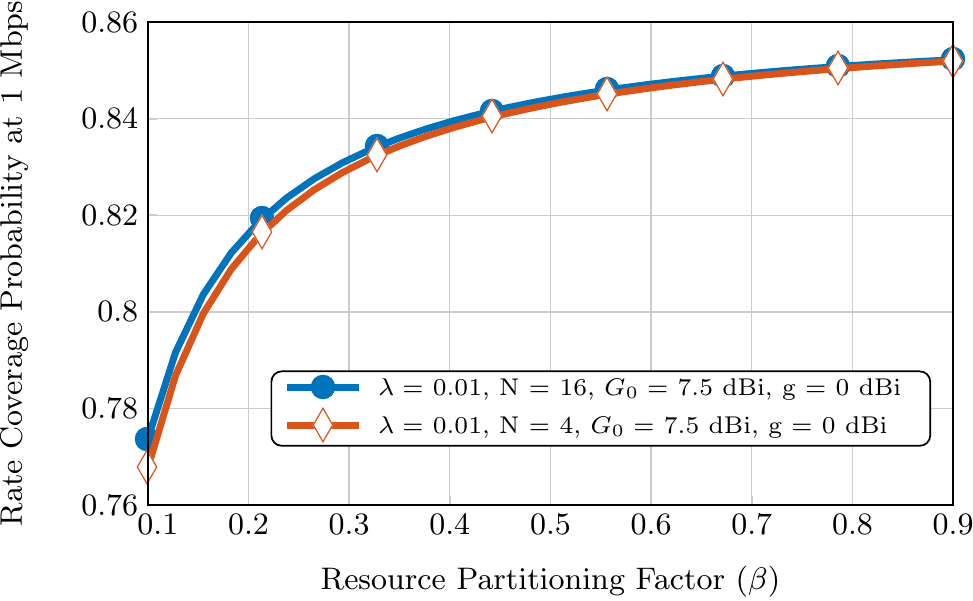}
\label{fig:Rate3}}
\caption{Rate coverage probability versus the resource partitioning factor for different beam dictionary sizes.}
%{[BD: (i) Can't we have the same order while exposing the parameters in the legend. It's not so easy to track the main differences btw. (a) and (b) in terms of settings at very first sight ? (ii) Why stating $G$ on Fig. (b) ?]}}
\vspace*{-1cm} \end{figure}
% \begin{figure}
% \centering
% \subfloat[]
% {\includegraphics[width = 8cm]{c3.eps}
% \label{fig:Rate3}}
% % \hfill
% % \subfloat[]
% % {\includegraphics[width = 8cm]{c4.eps}
% % \label{fig:Rate4}}
% \caption{Rate coverage probability versus the resource partitioning factor for different beam dictionary sizes. \textcolor{magenta}{MC: merge the 3 figures in 1 figure and 3 subfigures and add captions. In Fig10, beta should decrease at some moment, no?}}
% \vspace*{-1cm} \end{figure}
It must be noted that the rate coverage performance does not only depend on the antenna gains and $\lambda$, but also on the measurement noise. We {study this point} in the following sub-section.
% \begin{figure}
% \centering
% \includegraphics[width = 8cm]{Fig6.pdf}
% \caption{{Rate coverage probability vs resource partitioning factor for different beam dictionary sizes.}}
% \label{fig:Rate3}
% \vspace*{-1cm} \end{figure}

\begin{figure}
\centering
\subfloat[]
{\includegraphics[width =8cm, height = 4cm]{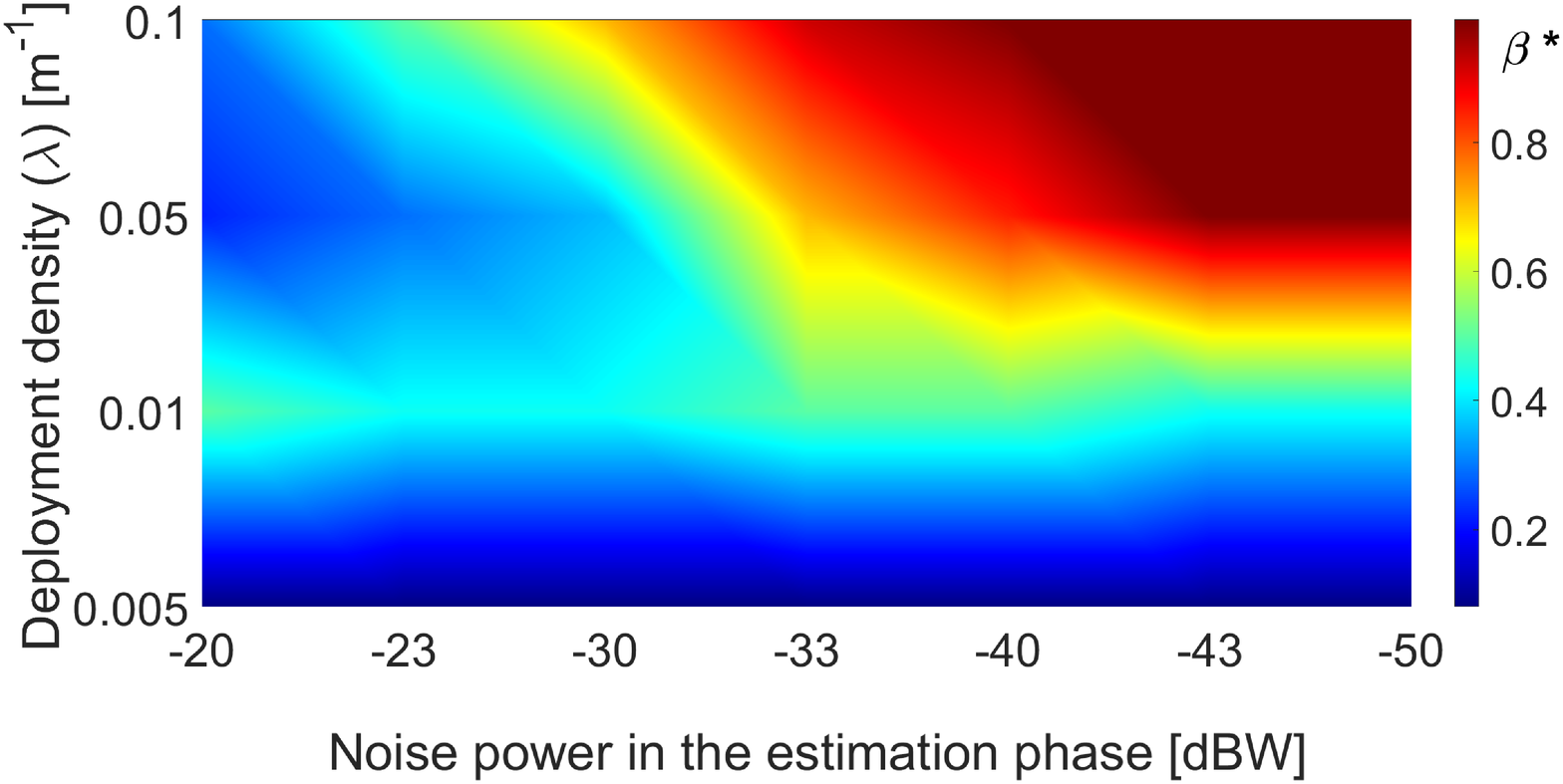}
\label{fig:Opt_beta}}
\hfill
\subfloat[]
{\includegraphics[width =8cm, height = 4cm]{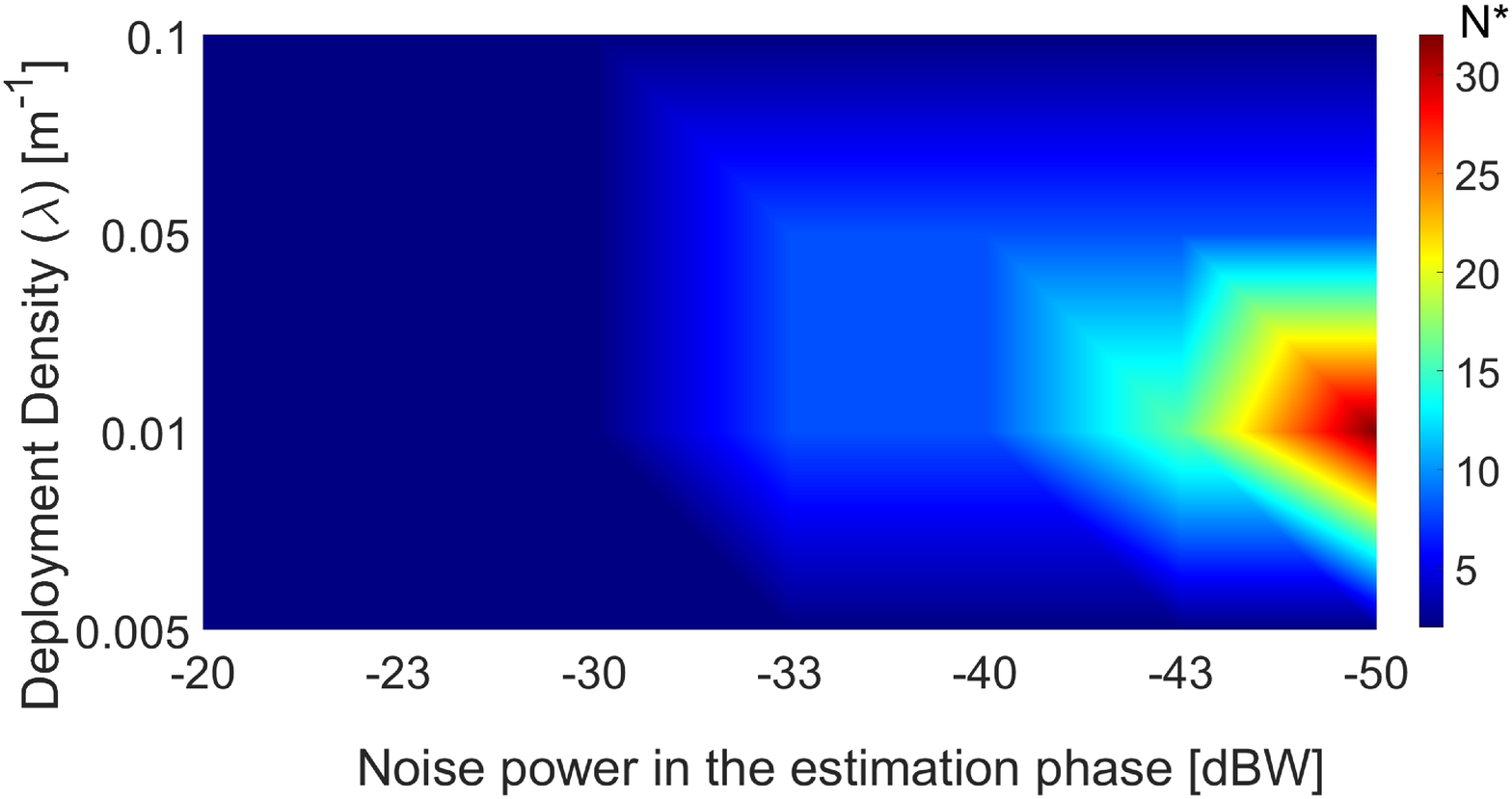}
\label{fig:Opt_N}}
\caption{(a) Optimal value of $\beta$ with respect to deployment density and noise; (b) Optimal beam-dictionary size with respect to deployment density and the noise.}
\vspace*{-1cm} \end{figure}

\vspace*{-0.5cm}
\subsection{{Optimal Partition factor and Beam Dictionary Size}}
%\textcolor{magenta}{MC: I would say influence of the bandwidth as $N=N_0W$ and I would vary $W$ bw typical values of mmW. Also: $N$ is not an estimation noise, it is thermal noise.}
{In this Section, we discuss the results obtained solving the transmit beamwidth and radio frame structure problem presented in \eqref{initial}.}
In Fig.~\ref{fig:Opt_beta} we plot the optimal values of $\beta$ with respect to the \ac{BS} deployment density $\lambda$ and the {noise power {$N_0 B$} [dBW]}. {For the optimization problem, we have considered $\epsilon = \epsilon' = 0.1$}. 
%{[(i) Shouldn't we specify more explicitly which "estimation" we're talking about, e.g., by making a reference to the adequate Eq.; (ii) please mind and clarify the exact meaning of $N_0$ and the related units accordingly, i.e., If it's "noise power", it should be expressed in dBm; if it's "noise power spectral density" it's expressed in dBm/Hz; if it's "SNR" it's expressed as a ratio in dB...]}
With low noise power (e.g., -50 dBW) the optimal value of $\beta$ is closer to 1 for higher $\lambda$. This is due to the fact that for low estimation noise and densely deployed \acp{BS}, even a limited amount of resources allocated to the localization phase results in a good localization performance. Thus, the optimal solution is to allocate large resources to the data phase for enhancing the rate coverage. On the other hand, for sparsely deployed \acp{BS}, larger amount of resources are required for efficient localization and the value of $\beta$ decreases, even for the case of low estimation noise.

Interestingly, in the case of high estimation noise (e.g., {$N_0 B$} = -20 dBW), when increasing the small cell density, the optimal $\beta$ increases at first and then decreases. This is due to the fact that for dense deployment of \acp{BS}, in case of high noise power of estimation, the effect of the beam-selection error is notable due to the concurrent large interference (since interfering \acp{BS} are closer due to higher density). This requires a lower value of optimal $\beta$ to facilitate efficient localization and reduce localization errors. Thus for higher estimation noise, the behaviour of optimal $\beta$ is not monotonous with respect to the deployment density.

In Fig.~\ref{fig:Opt_N} we plot the optimal beam-dictionary size with respect to $\lambda$ and {$N_0 B$}. For high estimation noise power, large beams (i.e., smaller dictionaries) must be used so as that the beam-selection error is limited. In case the estimation noise is low (e.g. -50 dBW), the optimal size of the beam dictionary at first increases with the deployment density, due to the fact that larger antenna gains improve the rate coverage. However, after a certain point (i.e., for very dense deployments), the optimal beam-dictionary size decreases to limit the beam-selection errors, which would have a large impact on the user performance due to the concurrent high interference.

\section{Conclusion}
\label{sec:Con}
In this paper we studied a mm-wave system deployed along the roads of a city to support localization and communication services simultaneously. We have proposed a novel localization bound-assisted initial beam-selection method for the mobile users, which reduces the latency of initial access by upto 75\%. Then {the localization performance bounds have also been used to derive} the downlink data-rate of the network in a system supporting jointly the localization and communication services. Our results highlight that increasing the resources allocated to the localization functions may or may not enhance the user data-rate. As a result, the study of the optimal resource partitioning factor is non-trivial.
Consequently, we have {highlighted and explored the main} trends in the optimal resource partitioning factor and mm-wave beamwidth with respect to the rate coverage probability, with varying BS deployment density, antenna gain, and {estimation noise}. Finally, we provided several key system-design insights {and guidelines} based on our results. This will aid a network operator to cater to the outdoor mobile users, which are a key target for the first generation deployments of outdoor \ac{mm-wave} \ac{BS}.
% \vspace{-1cm}
% \section*{Acknowledgement}
% This work has been carried out in the frame of the SECREDAS project, which is partly funded by the European Commission (H2020 EU.2.1.1.7 ECSEL – GA 783119).
% \vspace{-1cm}
\appendices
   \vspace*{-0.3cm}

    \section{}
    \label{app:avg_BSerror}
    
%   First note that for a given beamwidth of $\theta_N$, all the terms $d_{{{ R}}_{jN}}$ and $d_{{{ L}}_{jN}}$ can be represented in terms of $d_{{ a}}$ and $d_{{{ L}}_{1N}}$, where $d_{{ a}}$ is the random variable representing the inter-BS distance by the equivalence: $d_{{{ R}}_{NN}} = d_{{ a}}$ \textcolor{red}{may this should be moved where used; here it is not clear}.
   
 {\it Proof of {Lemma}~\ref{theo:avg_BSerror}:} Beam-selection error occurs for a {user in coverage of the} beam $j$ of beamwidth $\theta_k$ when the {estimated position} lies outside the {beam $j$}.
    Thus, the probability of beam-selection error, in case the user is estimated to be located at $\hat{d}$, is computed as: 
    \begin{align}
\mathcal{P}_{{ BS},{j,k}}\left(d,\sigma^2_d\right) &= \mathbb{P}\left({\hat{d}} < d_{{ L}_{jN}}\right) + \mathbb{P}\left({\hat{d}} > d_{{ R}_{jN}}\right) {\overset{(a)}{=} 1 - \mathcal{Q}\left(\frac{d_{{ L}_{jN}} - {d}}{{\sigma_d}}\right) + \mathcal{Q}\left(\frac{d_{{ R}_{jN}} - {d}}{{\sigma_d}}\right)}\nonumber,
\end{align}
{where (a) is due to the Gaussian nature of the error around mean $d$ and variance $\sigma_d^2$.}
Then, the probability of beam-selection error and the typical {user} is in the coverage area of the $j$-{th} beam:
{$$
\bar{\mathcal{P}}_{{ BS},j,k} = \int_{d_{L_{jN}}}^{d_{R_{jN}}} \mathcal{P}_{{ BS}, j,N}(x) f_{d}(x) dx.
$$}
Finally, the {average} beam-selection error for the localization based beam-selection scheme with a beam-dictionary size of $N$ is calculated as: $
\bar{\mathcal{P}}_{{ BS}} = \mathbb{E}_{d_a}\left[\sum_{j = 1}^{N(d_a)}\bar{\mathcal{P}}_{{ BS},j,k} \right]$.

\section{}
\label{app:CovP_LOS}
{\it Proof of Theorem~\ref{theo:CovP_LOS}:} Let index $1$ denote the serving BS and $z_k=\sqrt{d^2_k+h^2_B}$ the distance between the $k$-th BS and the typical UE. The probability that the SINR at the typical user is larger than a threshold $T$, in case of absence of beam-selection error and misalignment error is:
\begin{subequations}
\small{
\begin{align}
&\mathcal{T}_0 =\mathbb{P}\left(SINR_C\geq T\right)= \mathbb{P}\left(\frac{P_t K \Gamma z_1^{-\alpha_L} |f_1|^2}{{N_0} + P_t K g^2 \left(\sum_{i \in\xi_L \backslash \{1\}} z_i^{-\alpha_L}|f_i|^2 + \sum_{j \in\xi_N} z_j^{-\alpha_N}|f_j|^2\right)} \geq T\right) \nonumber \\
& = \mathbb{P}\left(|f_1|^2 \geq \frac{{T} {N_0} + P_t K g^2 \left(\sum_{i \in\xi_L \backslash \{1\}} z_i^{-\alpha_L}|f_i|^2 + \sum_{j \in\xi_N} z_j^{-\alpha_N}|f_j|^2\right)}{P_t K \Gamma z_1^{-\alpha_L}}\right)\nonumber  \\
& = \sum\limits_{n = 1}^{N_L}(-1)^{n+1} \; \binom {N_L}n {\mathbb{E}}\left[\exp\left(-\frac{n \eta_L T {N_0}}{P_t K \Gamma z_1{^{{-\alpha_L}}}} - \frac{n \eta_LTg^2\sum_{i \in\xi_L \backslash \{1\}} z_i^{-\alpha_L}|f_i|^2}{\Gamma z_1^{-\alpha_L} } - \frac{n \eta_L Tg^2\sum_{j \in\xi_N} z_j^{-\alpha_N}|f_j|^2}{\Gamma z_1^{-\alpha_L} }\right)\right] \nonumber \\
& = \sum\limits_{n = 1}^{N_L}(-1)^{n+1} \; \binom {N_L}n \exp\left(-\frac{n \eta_L T {N_0} }{P_t K \Gamma z_1{^{{-\alpha_L}}}}\right)\mathbb{E}_{|f_i|^2,\xi_L\backslash\{{1}\}}\left[\exp\left(- \frac{n \eta_L Tg^2\sum_{i \in\xi_L \backslash \{1\}} z_i^{-\alpha_L}|f_i|^2}{\Gamma z_1^{-\alpha_L} }\right)\right] \nonumber \\ &\hspace{5cm}\mathbb{E}_{|f_j|^2,\xi_N}\left[\exp\left(- \frac{n \eta_L Tg^2\sum_{j \in\xi_N} z_j^{-\alpha_N}|f_j|^2}{\Gamma z_1^{-\alpha_L} }\right)\right],  \nonumber
\end{align}
}
\end{subequations}
where $\Gamma = \gamma_{{ B}}(\theta_{{ B}})\gamma_{{ U}}(\theta_{{ U}})$ and $\eta_L = N_L (N_L!)^{-\frac{1}{N_L}}$~\cite{bai2015coverage}. Now,
\begin{subequations}
\small{
\begin{align}
&\mathbb{E}_{|f_i|^2,\xi_L\backslash\{{1}\}}\left[\exp\left(- \frac{n \eta_L Tg^2\sum_{i \in\xi_L \backslash \{1\}} z_i^{-\alpha_L}|f_i|^2}{ \Gamma z_1^{-\alpha_L} }\right)\right]\nonumber   = \mathbb{E}\left[\underset{i \in\xi_L \backslash \{1\}}{\prod} \mathbb{E}_{|f_i|^2}\left[\exp\left(- \frac{n \eta_L Tg^2 z_i^{-\alpha_L}|f_i|^2}{ \Gamma z_1^{-\alpha_L} }\right)\right]\right] \nonumber \\
& = \exp\left(\int_{d_1}^{d_S} 1 - \mathbb{E}_{|f_i|^2}\left[\exp\left(- \frac{n \eta_L Tg^2 (x^2 + h_B^2)^{\frac{-\alpha_L}{2}}|f_i|^2}{ \Gamma z_1^{-\alpha_L} }\right)\right]2\lambda x dx\right) \nonumber\\
&=  \exp\left(-2\lambda \int_{d_1}^{d_S}1- \frac{1}{\left( \frac{\eta_L Tg^2 (x^2 + h_B^2)^{\frac{-\alpha_L}{2}}|f_i|^2}{ N_L \Gamma z_1^{-\alpha_L} }\right)^{N_L}}x dx\right) \nonumber.
\end{align}
}
\end{subequations}
The NLOS case follows similarly.
% \begin{subequations}
% \begin{align}
% \scalemath{.9}{\mathbb{E}_{|f_j|^2,\xi_N}\left[\exp\left(- \frac{ Tg^2\sum_{j \in\xi_N} z_j^{-\alpha_N}|f_j|^2}{ \Gamma z_1^{-\alpha_L} }\right)\right] = \exp\left(-\int_{d_S}^{\infty}\frac{Tg^2(x^2 + h_B^2)^{\frac{-\alpha_N}{2}}}{ \Gamma z_1^{-\alpha_L} + Tg^2(x^2 + h_B^2)^{\frac{-\alpha_N}{2}}} 2\lambda x dx\right)\nonumber.}
% \end{align}
% \end{subequations}
To calculate $\mathcal{T}_{BS}$ and $\mathcal{T}_{MA}$, respectively in the events of beam-selection error and misalignment, we replace the values of $ \gamma_{{ B}}(\theta_{{ B}})$ and $\gamma_{{ U}}(\theta_{{ U}})$ with $g$ according to \eqref{eq:Gain}. Then, from the theorem of total probability, the SINR coverage at a distance $d_{1}$ is calculated. Conditioning on $d_1$ lying between $d_{L_i}$ and $d_{R_i}$ completes the proof.
\vspace{-0.3cm}
\bibliographystyle{IEEEtran}
	\bibliography{refer.bib}

\end{document}